\documentclass{article}

\newcommand{\quotes}[1]{``#1''}

\usepackage{xcolor}

\usepackage[noend]{algpseudocode}
\usepackage{algorithm,algorithmicx}
\newcommand*\Let[2]{\State #1 $\gets$ #2}

\algnewcommand\algorithmicupon{\textbf{upon}}
\algblockdefx[UPON]{Upon}{EndUpon}[1]
{\algorithmicupon\ #1\ \algorithmicdo}
{\algorithmicend\ \algorithmicupon}
\makeatletter
\ifthenelse{\equal{\ALG@noend}{t}}%
{\algtext*{EndUpon}}
{}%
\makeatother

\usepackage{etoolbox}
\usepackage{tikz}
\usetikzlibrary{tikzmark}
\usetikzlibrary{calc}
\newcommand{\ALGtikzmarkcolor}{black}
\newcommand{\ALGtikzmarkextraindent}{4pt}
\newcommand{\ALGtikzmarkverticaloffsetstart}{-.5ex}
\newcommand{\ALGtikzmarkverticaloffsetend}{-.5ex}
\makeatletter
\newcounter{ALG@tikzmark@tempcnta}

\newcommand\ALG@tikzmark@start{%
	\global\let\ALG@tikzmark@last\ALG@tikzmark@starttext%
	\expandafter\edef\csname ALG@tikzmark@\theALG@nested\endcsname{\theALG@tikzmark@tempcnta}%
	\tikzmark{ALG@tikzmark@start@\csname ALG@tikzmark@\theALG@nested\endcsname}%
	\addtocounter{ALG@tikzmark@tempcnta}{1}%
}

\def\ALG@tikzmark@starttext{start}
\newcommand\ALG@tikzmark@end{%
	\ifx\ALG@tikzmark@last\ALG@tikzmark@starttext
	\else
	\tikzmark{ALG@tikzmark@end@\csname ALG@tikzmark@\theALG@nested\endcsname}%
	\tikz[overlay,remember picture] \draw[\ALGtikzmarkcolor] let \p{S}=($(pic cs:ALG@tikzmark@start@\csname ALG@tikzmark@\theALG@nested\endcsname)+(\ALGtikzmarkextraindent,\ALGtikzmarkverticaloffsetstart)$), \p{E}=($(pic cs:ALG@tikzmark@end@\csname ALG@tikzmark@\theALG@nested\endcsname)+(\ALGtikzmarkextraindent,\ALGtikzmarkverticaloffsetend)$) in (\x{S},\y{S})--(\x{S},\y{E});%
	\fi
	\gdef\ALG@tikzmark@last{end}%
}

\apptocmd{\ALG@beginblock}{\ALG@tikzmark@start}{}{\errmessage{failed to patch}}
\pretocmd{\ALG@endblock}{\ALG@tikzmark@end}{}{\errmessage{failed to patch}}
\makeatother

\usepackage{marginnote}

\usepackage[title]{appendix}

\usepackage{soul}
\usepackage{amsmath}
\usepackage{amssymb}

\usepackage{mathtools}  
\DeclarePairedDelimiter\ceil{\lceil}{\rceil}

\usepackage{pifont}
\newcommand{\cmark}{\ding{51}}%
\newcommand{\xmark}{\ding{55}}%

\usepackage{authblk}

\usepackage{amsfonts}
\usepackage{hyperref}
\usepackage{subcaption}
\usepackage{multirow}

\usepackage{amsthm}

\newtheorem{theorem}{Theorem}
\newtheorem{lemma}[theorem]{Lemma}
\newtheorem{corollary}[theorem]{Corollary}
\newtheorem{definition}[theorem]{Definition}
\newtheorem{remark}[theorem]{Remark}

\title{Reliable Broadcast despite Mobile Byzantine Faults \footnote{Part of the content of this paper has been published in \cite{DBLP:conf/opodis/opodis23}}}

\author[1]{Silvia Bonomi}
\author[1]{Giovanni Farina}
\author[2]{Sébastien Tixeuil}
\affil[1]{Sapienza University of Rome, Rome, Italy}
\affil[2]{Sorbonne Université, CNRS, LIP6, Institut Universitaire de France, Paris, France}

\begin{document}
\date{}
\sloppy

\maketitle

\begin{abstract}
    We investigate the solvability of the Byzantine Reliable Broadcast and Byzantine Broadcast Channel problems in distributed systems affected by Mobile Byzantine Faults. 
    We show that both problems are not solvable even in one of the most constrained system models for mobile Byzantine faults defined so far.
    By endowing processes with an additional local failure oracle, we provide a solution to the Byzantine Broadcast Channel problem.
\end{abstract}

\section{Introduction}
Byzantine Reliable Broadcast (BRB) is a fundamental primitive in fault-tolerant distributed systems ensuring that all correct processes eventually deliver the same message from a defined sender regardless of its  correctness.
Defined by Bracha~\cite{DBLP:journals/iandc/Bracha87} as a building block for a Byzantine-tolerant consensus protocol, BRB has been widely adopted and investigated since then, thanks to its ability to prevent arbitrarily (i.e., Byzantine) faulty processes from \emph{equivocating} by sending different messages to different processes.
It has been introduced as a \emph{one-shot} primitive that allows a pre-defined process in the system to spread a single message and generalized as a Byzantine Broadcast Channel (BBC) primitive \cite{DBLP:books/daglib/0025983} to allow every process to spread an arbitrary number of messages.
BRB has been used to construct several fault-tolerant distributed solutions, solving more complex problems such as register abstractions, consensus problems, and distributed ledgers. Thus, it has been analyzed in the literature from various perspectives, such as minimizing bandwidth consumption \cite{DBLP:conf/podc/AlhaddadDD0VXZ22}, or latency \cite{DBLP:journals/ppl/ImbsR16, DBLP:conf/opodis/Abraham0X21}.

A fundamental perspective to consider is the investigation of the feasibility of BRB and BBC when assuming no permanent failures.
In this paper, we are interested in analyzing BRB and BBC solvability considering a \textit{dynamic process failure model}, i.e., a model in which every process may potentially fail and recover, causing a potentially continuous change in a process's failure state throughout the system's lifetime.
Some examples of systems considering dynamic process failures are crash-recovery systems \cite{DBLP:conf/icdcs/RodriguesR00,DBLP:journals/jpdc/BoichatG05}, self-stabilizing systems \cite{DBLP:journals/cacm/Dijkstra74,DBLP:books/mit/Dolev2000}, and Mobile Byzantine tolerant systems \cite{DBLP:conf/wdag/Garay94,DBLP:journals/tcs/BonnetDNP16}.
In this work, we consider the \textit{Mobile Byzantine Failure} (MBF) model, in which all processes may alternate between periods of correct behavior and periods of arbitrary behavior (i.e., Byzantine). Indeed, the failure state of processes is governed by an external attacker capable of compromising and controlling a set of processes in the system, and such a set is dynamic. 
The MBF model captures some of the features of the most frequent attacks targeting distributed systems and related countermeasures, where the process's faults are primarily due to external malicious causes rather than internal misbehavior, and tools such as software rejuvenation techniques \cite{DBLP:conf/issre/KoutrasP20}, intrusion detection systems \cite{DBLP:journals/jnca/LiaoLLT13}, and trusted execution environments \cite{DBLP:conf/trustcom/SabtAB15} are available.

Despite several fundamental distributed problems have been analyzed in the literature considering the MBF model (i.e., Byzantine agreement \cite{DBLP:conf/wdag/Garay94,DBLP:journals/tcs/BonnetDNP16}, approximate Byzantine agreement \cite{DBLP:conf/sss/Tseng17, DBLP:conf/icdcs/BonomiPPT16}, and registers emulation \cite{DBLP:conf/icdcn/BonomiPP16}), to the best of our knowledge the BRB problem has never been studied so far in such settings.

Thus, our objective in this paper is the investigation of BRB and BBC in the presence of MBFs.
In particular, our contributions are:
\begin{enumerate}
    \item we formalize the \emph{Mobile Byzantine Reliable Broadcast} (MBRB) and \emph{Mobile Byzantine Broadcast Channel} (MBBC) as a natural extension of the BRB and BBC specifications to deal with MBFs. Indeed, the standard specifications for BRB and BBC primitives consider a \textit{static failure model}, where every process is either permanently correct or faulty;
    \item we prove several impossibility results, mainly showing that MBRB and MBBC cannot be implemented without additional knowledge provided by a powerful oracle reporting about processes' failure state; 
    \item we introduce such a powerful oracle and provide a protocol for solving MBBC in a synchronous round-based system; 
    \item we analyze a weaker MBBC specification that can be realized without the oracle.
\end{enumerate}

Let us note that being a natural extension of BRB and BBC primitives, the MBRB and MBBC primitives prevent faulty processes from equivocating, namely from sending different information to different processes, and can be used as building block for other fault-tolerant primitives. For example, MBRB/MBBC primitives can extend mobile Byzantine fault-tolerant register abstractions to support Byzantine clients \cite{DBLP:conf/icdcn/BonomiPP16}.
Our work not only offers an analysis of a specific problem but also provides several insights for other distributed system problems where the failure state of a process is dynamic and partially or entirely unknown.
We consider relatively strong assumptions in our system model, the same as those considered in related work, in order determine fundamental solvability conditions. Relaxation of most of these assumptions has already been partially investigated \cite{DBLP:conf/srds/BonomiPPT17}.

The rest of the paper is structured as follows. After reviewing related work on implementations of the BRB primitive and contributions considering mobile Byzantine failures in Section \ref{sec:related}, we formalize the system model in Section \ref{sec:sysmod}. We introduce the new specifications for the Mobile Byzantine Reliable Broadcast and the Mobile Byzantine Broadcast Channel problems in Section \ref{sec:mbrb}.
Section \ref{sec:imp} presents some impossibilities for the specifications we defined. 
To overcome some of the identified impossibilities and solve the Mobile Byzantine Broadcast Channel problem, we consider a powerful oracle, we propose a protocol in Section \ref{sec:prot}, and we analyze a weaker Mobile Byzantine Broadcast Channel specification that is realizable without any oracle in Section \ref{sec:weakprim}.

\section{Related Work}
\label{sec:related}
The Byzantine Reliable Broadcast (BRB) abstraction has been introduced by Bracha \cite{DBLP:journals/iandc/Bracha87} as a building block for a Byzantine-tolerant consensus protocol in a distributed system where at most $f$ processes are permanently arbitrary (Byzantine) faulty. Thanks to its ability to guarantee agreement among correct processes over the set of delivered messages, a BRB primitive has been used as a building block from several fault-tolerant solutions, and has been intensely investigated under several system and failure models, with the final aim of extending its power and optimizing different performance metrics.\\
Imbs and Raynal \cite{DBLP:journals/ppl/ImbsR16} proposed a protocol that improves latency (in terms of the number of rounds of message exchanges) compared to Bracha.
Guerraoui et al. \cite{DBLP:conf/wdag/GuerraouiKMPS19} relaxed the BRB specification, allowing each property to be violated with a fixed and arbitrarily small probability. 
Backes and Cachin \cite{DBLP:conf/dsn/BackesC03} and Raynal \cite{DBLP:journals/ppl/Raynal21} discussed extensions of the BRB problem; the former assuming both Byzantine faulty processes and fail-stop failures, the latter distinguishing between two different kinds of Byzantine behaviors, i.e. those attempting to prevent the liveness and those attempting to prevent the safety of the BRB.
Recently, Guerraoui et al. \cite{DBLP:conf/opodis/GuerraouiKKPST20} and Li et al. \cite{DBLP:journals/access/LiYWW22} extended BRB to distributed systems with dynamic membership: in any given view (i.e. set of participating processes, governed by the processes themselves), the set of Byzantine processes remains the same; however, two consecutive views allow for different sets of Byzantine processes.
By contrast, our work considers a static system membership (i.e., a fixed set of processes participating in the protocol) but a dynamic failure model, where Byzantine processes may change (that is, recover, and get Byzantine again) during the \emph{same} view.
To the best of our knowledge, all existing BRB protocols that assumed arbitrary process failures, except the aforementioned works by Guerraoui et al.~\cite{DBLP:conf/opodis/GuerraouiKKPST20} and Li et al.~\cite{DBLP:journals/access/LiYWW22}, considered a \textit{static failure model} i.e., they assumed that the set of Byzantine processes does not change.

Mobile Byzantine Failure (MBF) models have been introduced to capture various types of faults, such as external attacks, virus infections, or even arbitrary behaviors caused by software bugs, using a single model encompassing detection and rejuvenation capabilities.
In all these models, failures are abstracted by an omniscient adversary that can control up to $f$ mobile Byzantine agents. Every agent is located in a process and makes it Byzantine faulty until the omniscient adversary decides to move it to another process.
The main differences between existing MBF models are in the power of the omniscient adversary (i.e., when it can move the agents) and in the awareness that every process has about its failure state.
Most MBF models considered \textit{round-based computations} and can be classified according to Byzantine mobility constraints: under \emph{constrained mobility} \cite{DBLP:conf/ftcs/BuhrmanGH95} the adversary can move agents only when protocol messages are sent (similarly to how viruses would propagate), while under \emph{unconstrained mobility} \cite{banu2012improved,DBLP:journals/tcs/BonnetDNP16, DBLP:conf/wdag/Garay94,DBLP:conf/podc/OstrovskyY91, DBLP:conf/opodis/SasakiYKY13, DBLP:journals/iandc/Reischuk85} agents do not move with messages but rather during specific phases of the round.
More in detail, Reischuk \cite{DBLP:journals/iandc/Reischuk85} considered malicious agents stationary for a given period; Ostrovsky and Yung \cite{DBLP:conf/podc/OstrovskyY91} introduced the notion of mobile viruses and defined the adversary as an entity that can inject and distribute faults; finally, Garay \cite{DBLP:conf/wdag/Garay94}, Banu et al. \cite{banu2012improved}, Sasaki et al. \cite{DBLP:conf/opodis/SasakiYKY13}, and Bonnet et al. \cite{DBLP:journals/tcs/BonnetDNP16} considered that processes execute synchronous rounds and mobile agents can move from one process to another in a specific phase of the round, which subsequently affects each process's ability to adhere to the algorithm. 
As a result, the set of Byzantine faulty processes at any given moment is limited in size; however, its composition may change from one round to the next, and the impact of past compromises may linger if not properly addressed by the protocol.
The aforementioned works \cite{DBLP:conf/wdag/Garay94,banu2012improved,DBLP:conf/opodis/SasakiYKY13,DBLP:journals/tcs/BonnetDNP16} also differ due to the assumption about the knowledge that processes have about their previous infection.
In the Garay model \cite{DBLP:conf/wdag/Garay94}, a process can detect its infection after the agent leaves it. Conversely, Sasaki et al. \cite{DBLP:conf/opodis/SasakiYKY13} investigated a model where processes cannot detect when agents leave. Finally, Bonnet et al. \cite{DBLP:journals/tcs/BonnetDNP16} considered an intermediate setting where not faulty processes control the messages they send (in particular, they send the same message to all destinations, and they do not send spurious information). 
Bonomi et al. \cite{DBLP:conf/podc/BonomiPPT16, DBLP:conf/srds/BonomiPPT17} decoupled algorithm rounds from Mobile Byzantine agent movement (\textit{round-free model}). 
The problems analyzed under MBF models are Byzantine agreement \cite{DBLP:conf/wdag/Garay94,banu2012improved,DBLP:conf/opodis/SasakiYKY13,DBLP:journals/tcs/BonnetDNP16}, approximate Byzantine agreement \cite{DBLP:conf/sss/Tseng17, DBLP:conf/nca/SakavalasT18, DBLP:conf/icdcs/BonomiPPT16}, and Byzantine-tolerant registers \cite{DBLP:conf/podc/BonomiPPT16, DBLP:conf/icdcn/BonomiPP16, DBLP:conf/srds/BonomiPPT17}. To the best of our knowledge, no efforts have been made to investigate the BRB problem in the presence of MBFs. All existing works that assume MBFs rely on some kind of best-effort communication subsystem (i.e., no guarantees exist when a process is controlled by a Mobile Byzantine agent), potential equivocations and omissions introduced by faulty processes are directly addressed by the main investigated primitive (e.g., consensus, register). {The existence of a BRB primitive can simplify the definition of other mobile Byzantine fault-tolerant primitives, similar to the case of the static failure model~\cite{DBLP:journals/iandc/Bracha87}.}

\section{System Model}
\label{sec:sysmod}
We consider a distributed system composed of a set of $n$ processes $\Pi=\{p_1, p_2 \dots p_n\}$, each associated with a unique identifier. 

Processes communicate through message passing.
We assume that a process can communicate with any other process through a \textit{reliable}, \textit{authenticated}, \textit{point-to-point link} abstraction \cite{DBLP:books/daglib/0025983}. This means that messages sent over such channels cannot be altered, dropped, or duplicated, and the identity of the sender cannot be forged. 
A reliable authenticated point-to-point link abstraction exposes two operations: \emph{(i)} ${\sf P2P.send}(p_{\mathit{rcv}},m)$ which sends the message $m$ to the receiver process $p_{\mathit{rcv}}$, and \emph{(ii)} ${\sf P2P.deliver}(p_{\mathit{snd}},m)$ which notifies the reception of the message $m$ from a sender process $p_{\mathit{snd}}$.

We measure the time according to a fictional global clock $\mathbb{T}$ (not accessible to processes) spanning over the set of natural numbers $\mathbb{N}$. We refer to the starting time of the system as $t_0$, the $i$-th time instant since the beginning of the execution as $t_i$, and a period of time between time $t_b$ and $t_e$ as $T_{b,e} := [t_b, t_e) : t_b, t_e \in \mathbb{T}; t_b < t_e$.

Each process executes a distributed protocol $\mathcal{P}$ consisting of a set of local algorithms. Each algorithm in $\mathcal{P}$ is represented by a finite state automaton whose transitions correspond to computation and communication steps. A computation step denotes a computation executed locally by a given process, while a communication step denotes the sending or receiving of a message. Computation steps and communication steps are generally called \emph{events}.
Each process maintains a set of variables. This set and the current value of those variables denote the \emph{state} of a process. 

\begin{definition}[Local Execution History]
A local execution history is an alternating sequence $s_0,e_0,s_1,e_1,\ldots$ of states and events of a process $p_i$, such that state $s_{j+1}$ results from state $s_j$ by executing event $e_j$.
\end{definition}

\noindent We assume that the local algorithms composing $\mathcal{P}$ are stored in a tamper-proof read-only memory.

Processes may fail and we assume that they are affected by \emph{Mobile Byzantine Failures} (MBF). 
That is, we assume the existence of an omniscient adversary that controls up to $f>0$ mobile Byzantine agents and that can ``move'' such agents from one set of processes to another. 
When the adversary places a Byzantine agent on a process $p_i$, the agent takes control of $p_i$, letting it behave arbitrarily. For example, $p_i$ may omit to send/receive messages, alter the content of messages, alter its process state regardless of its local algorithm, and execute arbitrary code. However, we assume that the mobile Byzantine agents cannot compromise the code stored in the tamper-proof memory.
Thus, when the Byzantine agent leaves $p_i$, $p_i$ resumes executing its local algorithm correctly (albeit from a possibly corrupted state).
We assume that the adversary can move each mobile agent independently of the others. Still, any agent must remain on a process for a period of time lasting at least $\Delta_s \in \mathbb{Q}^{+}$ (rational positive numbers), i.e., once arrived, an agent compromises a node for at least $\Delta_s$ consecutive time units, and when $\Delta_s < 1$ we have that an agent can move multiple times in the same time unit. 
As an example, if $\Delta_s = 2$ we have that every mobile Byzantine agent must remain on the same process for at least $2$ consecutive time units, while $\Delta_s = \frac{1}{2}$ means that the agent may move {$\ceil*{\frac{1}{\Delta_s}} = 2$ times} in a time unit and compromise {$\ceil*{\frac{1}{\Delta_s}} =2$} different processes in the same time unit.\\
Let us note that, in the MBF model, no single process is guaranteed to remain correct forever and we may have processes that alternate between correct and incorrect behavior infinitely often.
This fundamental difference from the classical static Byzantine failure model commands to redefine the notion of correct and faulty processes (i.e., \textit{the process failure states}).

\begin{definition}[Faulty process]
A process $p_i$ is said to be \emph{faulty} at time $t_k$ if it is controlled by a mobile Byzantine agent at time $t_k$.
By extension, if at each time between $t_b$ and $t_e$, process $p_i$ is \emph{faulty}, then $p_i$ is faulty during the period $T_{b, e}$. 
\end{definition}

\noindent When a process $p_i$ is faulty, it may execute a protocol $\mathcal{P}'\neq \mathcal{P}$, and its local state may be altered arbitrarily.

\noindent We denote by $B(t)$ the set of faulty processes at time $t$ and by $B(T_{b,e})$ the set of faulty processes during the whole period $T_{b,e}$ (i.e., $B(T_{b,e}) = \bigcap_i B(t_i)$ for $b \le i < e$). 

\begin{definition}[Correct process]
A process $p_i$ is \emph{correct} when it is not faulty, that is, $p_i$ is correct at time $t_k$ if it is not controlled by a Byzantine agent at time $t_k$. Similarly, a process $p_i$ is correct in the period $T_{b, e}$ if it remains correct between times $t_b$ and $t_e$.
\end{definition}

\noindent Let us remark that when a process $p_i$ is correct, it executes $\mathcal{P}$ but potentially it may start its execution from a compromised state (due to a previous corruption performed by a mobile Byzantine agent). 
\noindent We denote by $C(t_k)$ the set of correct processes at time $t_k$ and by $C(T_{b,e})$ the set of correct processes throughout the period $T_{b,e}$ (that is, $C(T_{b,e}) = \bigcap_i C(t_i)$ for $b \le i < e$). 

Note that, due to the mobility of Byzantine agents, every process may potentially alternate between correct and faulty states infinitely often. 
To this aim, we also introduce the notion of \emph{infinitely often correct processes}:

\begin{definition}[$\Delta_c$-Infinitely often correct process]
Let $\Delta_c \in \mathbb{N}^{+}$.
A process $p_i$ is $\Delta_c$-\emph{infinitely often correct} if, for every time $t_j$, there exists a following period $T_{b,e}$ lasting at least $\Delta_c$ where $p_i$ is correct. Formally: $\forall t_j \in \mathbb{T}, ~ \exists t_b, t_e $ such that $t_b>t_j, ~t_e - t_b \geq \Delta_c, ~p_i \in C(T_{b,e})$.
\end{definition}

\noindent Informally, the notion of $\Delta_c$-infinitely often correct process captures the possibility that a process is not permanently faulty, but correct for at least $\Delta_c$ units of time after mobile Byzantine agents have left it.

\noindent In the following, we will consider several alternative settings for our system model:

\begin{itemize}
    \item \textit{\bf system timing assumptions}: we consider either a \textit{synchronous} (\texttt{SYNC}) or an \textit{asynchronous} (\texttt{ASYNC}) system. When considering a synchronous system, we assume that there is an upper bound on the time required to perform local computation on the processes and an upper bound on the time required by a message to be delivered via a P2P link, both of them known by all processes. 
    In addition, we assume that the computation evolves in sequential synchronous rounds $r_1, r_2, \dots ,r_j, \dots$. Every round $r_j$ is divided into three phases: \emph{(i)} \emph{send} where processes transmit messages to their intended receivers, \emph{(ii)} \emph{receive} where processes collect messages sent during the send phase of the current round, and \emph{(iii)} \emph{compute} where processes process received messages, and prepare those that need to be sent in the following round.
    Contrarily, in an asynchronous setting, we are not assuming any upper bound, and the computation progresses as soon as an event is generated by a process.
    
    \item \textit{\bf mobile Byzantine agent synchronization}: we consider three different types of mobility with different degrees of synchronization between mobile Byzantine agents. In particular, we will consider movement that are either \textit{synchronized} (\texttt{S-MOB$^{+}$}), \textit{synchronous} (\texttt{S-MOB}), or \textit{asynchronous} (\texttt{A-MOB}) that abstract MBF models existing in the literature. 
    In the \texttt{A-MOB} model, mobile Byzantine agents move independently and once the movement occurs, the agent remains at the destination node for at least $\Delta_s$, with $\Delta_s$ unknown to the processes (see ITU model in \cite{DBLP:conf/podc/BonomiPPT16}). 
    In the \texttt{S-MOB} model, mobile Byzantine agents move independently, and, also in this case, once the movement happens the agent remains on the destination node for at least $\Delta_s$. Unlike the previous case, $\Delta_s$ is known to the processes (see the ITB model in \cite{DBLP:conf/podc/BonomiPPT16}).
    The \texttt{S-MOB$^{+}$} model is a particular case of the \texttt{S-MOB} model specific for synchronous systems where the computation evolves in synchronous rounds. Indeed, in this case $\Delta_s$ is expressed in terms of round, and mobile Byzantine agents can move only between two consecutive rounds, i.e. after the computation phase of a round $r_i$ and before the send phase of round $r_{i+1}$\footnote{The agents' movements are thus synchronized with the synchronous rounds.}(see Garay's MBF model \cite{DBLP:conf/wdag/Garay94}). 
    Let us stress that in the \texttt{S-MOB$^{+}$} setting 
    every process is either faulty or correct for an entire round. Therefore, for ease of presentation, we say that a process is \emph{faulty or correct in the round $r_k$} in the \texttt{S-MOB$^{+}$} systems and extend the notation of $C(t)$ and $B(t)$ accordingly, that is, with $C(r_k)$ and $B(r_k)$, respectively, referring to the sets of correct and faulty processes in the round $r_k$. Furthermore, we measure the time with the number of rounds.
    
    \item \textit{\bf failure awareness}: we assume that every process $p_i$ is either \textit{aware} or \textit{unaware} about a mobile Byzantine agent moving away from $p_i$. 
    We abstract this knowledge by introducing two different local oracles that reveal information to process $p_i$. 
    Specifically, we consider: \textit{basic failure awareness} ($\mathcal{O}_{\mathit{BFA}}$) and \textit{full failure awareness} ($\mathcal{O}_{\mathit{FFA}}$).
    In the $\mathcal{O}_{\mathit{BFA}}$ case, a process $p_i$ knows when (i.e., in which time unit) a mobile agent moves away from $p_i$; 
    in the $\mathcal{O}_{\mathit{FFA}}$ case, a processes $p_i$ additionally know when the agent arrived to $p_i$ (i.e., $p_i$ know the entire period $T_{b, e}$ in which it was faulty). 
    \end{itemize}

\noindent More formally:
\begin{definition}[\textit{Basic Failure Awareness Oracle} $\mathcal{O}_{\mathit{BFA}}$]
    If a mobile Byzantine agent leaves from a process $p_i$ at time $t_j$, then the failure awareness oracle $\mathcal{O}_{\mathit{BFA}}$ generates a \textsc{cured}$()$ event on $p_i$ at time $t_{j+1}$.
\end{definition}

\noindent Observe that $\mathcal{O}_{\mathit{BFA}}$ informs $p_i$ as soon as $p_i$ becomes free from mobile Byzantine agents, and thus allows $p_i$ to take corrective actions (\emph{e.g.} to avoid spreading compromised information). However, $\mathcal{O}_{\mathit{BFA}}$ does not provide any information about the length of the period $p_i$ was faulty.

\begin{definition}[\textit{Full Failure Awareness Oracle} $\mathcal{O}_{\mathit{FFA}}$]
   If a mobile Byzantine agent takes control of a process $p_i$ at time $t_j$ and leaves $p_i$ at time $t_k$, then the full failure awareness oracle $\mathcal{O}_{\mathit{FFA}}$ generates a \textsc{cured}$()$ event on $p_i$ at time $t_{k+1}$, and returns the time label $t_j$ when invoking operation \textsc{faulty\_at}$()$.
\end{definition}

\noindent For the sake of notation, we refer to setting where no oracle is available as $\mathcal{O}_{\mathit{NFA}}$. Let us remark that both $\mathcal{O}_{\mathit{BFA}}$ and $\mathcal{O}_{\mathit{FFA}}$ are \emph{local} oracles, i.e., they provide information to the actual process where the events occurred; thus, a process $p_i$ is not aware of the failure state of any other process $p_j$.

Note that the assumptions considered in our system model are equivalent to or less constrained than those in other works dealing with mobile Byzantine agents ~\cite{DBLP:conf/wdag/Garay94,banu2012improved,DBLP:conf/opodis/SasakiYKY13,DBLP:journals/tcs/BonnetDNP16}. The only exceptions are the $\mathcal{O}_{\mathit{FFA}}$ oracle and the notion of $\Delta_c$-\emph{infinitely often correct} process, which have not been considered before.

In the remainder of the paper, we will characterize the specific setting considered in terms of system timing assumptions, agent synchronization, and failure awareness by specifying a triple $\langle \alpha, \beta, \gamma \rangle$ where $\alpha \in \{\texttt{SYNC}, \texttt{ASYNC}\}$, $\beta \in \{\texttt{A-MOB}, \texttt{S-MOB}, \texttt{S-MOB$^+$}\}$ and $\gamma \in \{\mathcal{O}_{\mathit{BFA}}, \mathcal{O}_{\mathit{FFA}}, \mathcal{O}_{\mathit{NFA}}\}$. With slight abuse of notation, we will use $"*"$ in a triple when the specific dimension is not relevant to prove our claims.

\section{Mobile BRB and BBC Specification}
\label{sec:mbrb}
\noindent Informally \textit{Byzantine Reliable Broadcast} (BRB) \cite{DBLP:journals/iandc/Bracha87,DBLP:books/daglib/0025983} is a communication primitive that enables all processes of a distributed system to agree on the delivery of a \underline{single} message disseminated by a pre-defined process called the \textit{source}, while the \textit{Byzantine Broadcast Channel} (BBC) \cite{DBLP:books/daglib/0025983} primitive extends BRB allowing all processes to disseminate \underline{an arbitrary number} of messages so that all correct processes eventually deliver the same set of messages \footnote{The formal specification of BRB and BBC primitives are provided in the Appendix \ref{sec:brb}.}.

Let us note that in the original BRB and BBC specifications the source is either always correct or always faulty in a given execution. Conversely, in our settings, it is possible that the source of a message changes its failure state multiple times (even during a single broadcast instance) making the original specification no more suitable. Thus, we extend the BRB and BBC, by formalizing the \textit{Mobile Byzantine Reliable Broadcast} (MBRB) and the \textit{Mobile Byzantine Broadcast Channel} (MBBC) problems to capture challenges imposed by mobile Byzantine faults.
We aim to specify two communication primitives accessible by every process and exposing the \textsc{{MBRB/MBBC.Broadcast($m$)}} and \textsc{MBRB/MBBC.Deliver($s$,$m$)} operations, where $m$ is a message and $s$ is a process identifier.
We say that a process $p_i$ \quotes{MBRB/MBBC-broadcasts a message $m$} when it executes \textsc{MBRB/MBBC.Broadcast($m$)}, and $p_i$ \quotes{MBRB/MBBC-delivers a message $m$ from $p_s$} when $p_i$ generates the \textsc{MBRB/MBBC.Deliver($s,m$)} event.
Similarly to other communication primitives, the \textsc{MBRB/MBBC-broadcast} operation is triggered to disseminate a message, while \textsc{MBRB/MBBC-deliver} notifies message deliveries.
We associate two additional parameters to both primitives, $\Delta_b \in \mathbb{N}^{+}$ and $\Delta_c \in \mathbb{N}^{+}$, characterizing the length of two periods (detailed in the specifications' properties).
We use the character \quotes{*} in our specifications when the actual value of the reference parameter is irrelevant.
\smallskip

\noindent Informally, a \textsc{MBRB}($\Delta_b,\Delta_c$) communication primitive guarantees that, given a source process $p_s$ and a message $m$ generated by $p_s$ while it is correct (for at least $\Delta_b$ time units), $m$ is reliably delivered by any $\Delta_c$-infinitely often correct process $p_j$ in a period where $p_j$ is correct. Similarly to BRB, this primitive is specified by considering an instance for every message generated by the identified source.
More formally, a \textsc{MBRB}($\Delta_b,\Delta_c$) communication primitive must guarantee the following properties:
\begin{itemize}
    \item \textit{$\left(\Delta_b,\Delta_c\right)$-Validity}: If there exists a period $T_{i,j}$ lasting at least $\Delta_b$ where a process $p_s$ is correct in $T_{i,j}$ and executes \textsc{MBRB.Broadcast}($m$), then at least one $\Delta_c$-infinitely often correct process $p_d$ eventually executes \textsc{MBRB.Deliver}($s$,$m$) while correct. 
    \item \textit{No duplication}: Every process $p_d$ executes \textsc{MBRB.Deliver}($s$,$*$) at most once when correct, namely $p_d$ MBRB-delivers at most one message from $p_s$ among all times $t_k \in \mathbb{T}$ such that $p_d \in C(T_{k, k+1})$.
    \item \textit{$\Delta_b$-Integrity}: If a process $p_d$ is correct at time $t_k$ and executes \textsc{MBRB.Deliver}($s$,$m$),
    then either $p_s$ was correct in $T_{i,j} = [t_i,t_{i+\Delta_b})$, with $t_i \leq t_k$, and executed \textsc{MBRB.Broadcast}($m$) at time $t_i$, or $p_s$ was faulty at some $t_i \leq t_k$.
    \item \textit{Consistency}: If some process is correct at time $t_k$ and executes \textsc{MBRB.Deliver}($s,m$), and another process is correct at time $t_l$ and executes \textsc{MBRB.Deliver}($s,m'$), then $m = m'$.
    \item \textit{$\Delta_c$-Totality}: If some process is correct at time $t_k$ and executes \textsc{MBRB.Deliver}($s,*$), then every $\Delta_c$-infinitely often correct process eventually executes \textsc{MBRB.Deliver}($s,*$).
\end{itemize}
~\\

The MBBC communication primitive is the natural extension of the BBC and its specification extends the one of the MBRB. In particular, the MBBC primitive 
guarantees that multiple messages generated by a source process (while it is correct for at least $\Delta_b$ consecutive time units) will be eventually delivered by any process $p_j$ that is $\Delta_c$-infinitely often correct in a period in which $p_j$ is correct.
More formally, a \textsc{MBBC}($\Delta_b$,$\Delta_c$) communication primitive must guarantee the following properties:
\begin{itemize}
    \item \textit{$\left(\Delta_b,\Delta_c\right)$-Validity}: If there exists a period $T_{i,j}$ lasting at least $\Delta_b$ where a process $p_s$ is correct in $T_{i,j}$ and executes \textsc{MBRB.Broadcast}($m$), then at least one $\Delta_c$-infinitely often correct process $p_d$ eventually executes \textsc{MBRB.Deliver}($s$,$m$) while correct. \\
    
    \item \textit{No duplication}: Every process $p_d$ executes \textsc{MBBC.Deliver}($s$,$m$), with message $m$ and source $s$, at most once when correct, namely, it MBBC-delivers a message $m$ from $p_s$ at most once among all times $t_k$ such that $p_d \in C(T_{k,k+1})$.\\

    \item \textit{$\Delta_b$-Integrity}: If a process $p_d$ is correct at time $t_k$ and executes \textsc{MBRB.Deliver}($s$,$m$),
    then either $p_s$ was correct in $T_{i,j} = [t_i,t_{i+\Delta_b})$, with $t_i \leq t_k$, and executed \textsc{MBRB.Broadcast}($m$) at time $t_i$, or $p_s$ was faulty at some $t_i \leq t_k$. \\

    \item \textit{$\Delta_c$-Agreement}: If some process is correct at time $t_k$ and executes \textsc{MBRB.Deliver}($s,m$), then every $\Delta_c$-infinitely often correct process eventually executes \textsc{MBRB.Deliver}($s,m$).
\end{itemize}

Note that the specifications rule the \textsc{MBRB/MBBC.Deliver}($s,m$) operations in times when processes \underline{are correct}. 
Operations executed when a process is faulty cannot be controlled and thus are not relevant to the specification.
Furthermore, note that when a process is controlled by a mobile Byzantine agent, it may execute arbitrary code and alter its local memory. Such a process has no information about what occurred when compromised (except the fact of being previously compromised in case an oracle is available).
This makes the implementation of the presented communication primitives particularly challenging and
will lead to proving several impossibility results that are specific to mobile Byzantine faults in the following sections.

\section{Impossibility Results}
\label{sec:imp}
This section presents several impossibility results for the MBRB and MBBC problems. 
In particular, Theorems \ref{th:impAsy} and \ref{th:impAsy2} prove the impossibility of solving both MBRB and MBBC if the system is asynchronous, or if the agents' movements are asynchronous. 
Then, assuming a synchronous system and synchronized agents,
Theorems \ref{th:mbrbimp1} and \ref{th:mbrbimp2} state the impossibility of solving MBRB with the strongest failure oracle we considered, $\mathcal{O}_{\mathit{FFA}}$, and the impossibility of solving MBBC with the weaker failure oracle, $\mathcal{O}_{\mathit{BFA}}$.
These latter impossibilities arise from the fact that a correct process cannot infer other processes' failure state from their behavior. Thus, they cannot distinguish messages that must be delivered from those that can be safely dropped.
Table \ref{tab:summaryImpossibilities} provides an overview of the impossibilities proved in this Section based on the specific considered settings.

    \begin{theorem}
        \label{th:impAsy}
        There exists no protocol $\mathcal{P}$ implementing the Mobile Byzantine Reliable Broadcast (resp. Mobile Byzantine Broadcast Channel) in $\langle \texttt{ASYNC}, \texttt{S-MOB}, \mathcal{O}_{\mathit{FFA}}\rangle$.
    \end{theorem}

   \begin{proof}
    \let\clearpage\relax 
%
In order to prove our claim we first show that it is impossible for any protocol $\mathcal{P}$ solving MBRB to generate an execution satisfying both \textit{$\left(\Delta_b,\Delta_c\right)$-Validity} and \textit{$\Delta_c$-Totality}. 
        Then we extend our arguments to prove the claim also for MBBC, where it is impossible to satisfy both \textit{$\left(\Delta_b,\Delta_c\right)$-Validity} and \textit{$\Delta_c$-Agreement}.
    
        Let us consider a process $p_s$ that is correct at a certain time $t_{bcast}$, that triggers \textsc{MBRB.broadcast($m$)} at time $t_{bcast}$, and that remains correct for a period $\Delta_{src} \geq \Delta_b$ after $t_{bcast}$.
        If $\mathcal{P}$ exists, it needs to guarantee \textit{$\left(\Delta_b,\Delta_c\right)$-Validity} for the message $m$. 
        As a consequence, if there exists a $\Delta_c$-infinitely often correct process $p_{dest}$, $\mathcal{P}$ must guarantee that eventually a \textsc{MBRB.deliver($s,m$)} event is generated from $p_{dest}$. 
        To guarantee both \textit{$\left(\Delta_b,\Delta_c\right)$-Validity} and \textit{$\Delta_c$-Totality} $p_{dest}$ must be different from $p_s$. It is therefore necessary that in $\mathcal{P}$ $p_{src}$ sends the message $m$ through the reliable authenticated links at least once to allow a $\Delta_c$-infinitely often correct process $p_{dest}$ to become aware of the message.
        
        Let us remark that under the \texttt{ASYNC} timing assumptions, there not exists any upper bound on the time required to exchange a message over a P2P link. 
        In particular, given a message $m$ sent by a process $p_i$ to a process $p_j$ at a certain time $t$ using a reliable authenticated perfect point-to-point link, we can only guarantee that $m$ will be delivered to $p_j$ at some time $t' > t$ but it is not possible to estimate its latency $d = t' - t$ (i.e., the time needed to deliver $m$).
        As a consequence, it is easily to identify a scenario where a single mobile Byzantine agent moves $n$ times in $T_{t, t'}$ and corrupts in sequence every processes $p_i$ in the system right after the message is P2P-delivered on $p_i$, discarding the message when received (mimic the loss of the message) and thus preventing a process $p_{dest} \neq p_s$ from delivering $m$.
    
        The reasoning can be extended considering many processes $p_i$ and for a MBBC instance with respect the properties \textit{$\left(\Delta_b,\Delta_c\right)$-Validity} and \textit{$\Delta_c$-Agreement}, and the claim follows.
    \end{proof}

Let us note that Theorem \ref{th:impAsy} holds assuming the most constrained agent's mobility model available in an asynchronous system (i.e., \texttt{S-MOB}) and the most powerful failure oracle ($\mathcal{O}_{\mathit{FFA}}$) considered. It follows that the MBRB and MBBC problems cannot be solved in \texttt{ASYNC} assuming a less constrained environment, as stated in the following Corollary.

    \begin{corollary}
        \label{cor:2}
        There exists no protocol $\mathcal{P}$ implementing the Mobile Byzantine Reliable Broadcast (resp. Mobile Byzantine Broadcast Channel) in 
        $\langle\texttt{ASYNC}, M, O\rangle$, 
        with $M\in\left\{\texttt{A-MOB},\texttt{S-MOB}\right\}$ 
        and $O\in\left\{\mathcal{O}_{\mathit{FFA}},\mathcal{O}_{\mathit{BFA}}\right\}$.
    \end{corollary}

    \begin{proof}
        \let\clearpage\relax 
%
The claim follows from the same argument provided as in Theorem \ref{th:impAsy}, given that the $\langle \texttt{ASYNC}, \texttt{S-MOB}, \mathcal{O}_{\mathit{FFA}}\rangle$ setting is the strongest possible for the parameters $M$ and $O$.
    \end{proof}

    \begin{theorem}
        \label{th:impAsy2}
        There exists no protocol $\mathcal{P}$ implementing the Mobile Byzantine Reliable Broadcast (resp. Mobile Byzantine Broadcast Channel) in $\langle \texttt{SYNC}, \texttt{A-MOB}, \mathcal{O}_{\mathit{FFA}} \rangle$.
    \end{theorem}

 \begin{proof}
        \let\clearpage\relax 
%
    The proof follow from Theorem \ref{th:impAsy} by observing that the same misbehavior occurs both in the $\langle \texttt{SYNC}, \texttt{A-MOB}, \mathcal{O}_{\mathit{FFA}} \rangle$ and $\langle \texttt{ASYNC}, \texttt{S-MOB}, \mathcal{O}_{\mathit{FFA}}\rangle$ settings. Indeed, if the latency of the communication $\delta$ is bounded due to the \texttt{SYNC} model, it is always possible to find a value for $\Delta_s < 1$ such that in every period $T_{t, t+\delta}$ a single mobile Byzantine agent can compromise $n$ different processes. Considering that in \texttt{A-MOB} processes do not know the value of $\Delta_s$ they cannot leverage on it in the protocol $\mathcal{P}$.
    \end{proof} 

\begin{theorem}
\label{th:mbrbimp1}
If $\Delta_b \in \mathbb{N}^{+}$ and $\Delta_b \geq 2$ rounds, then
there exists no protocol $\mathcal{P}$ implementing a Mobile Byzantine Reliable Broadcast primitive in $\langle$\texttt{SYNC}, \texttt{S-MOB}$^{+}$, $\mathcal{O}_{\mathit{FFA}}\rangle$.
\end{theorem}

\begin{proof}
For the sake of contradiction, let us assume that such a protocol $\mathcal{P}$ exists.
Let us consider the local execution history $\mathcal{H}'_s$ of a process $p_s$ that is \textit{correct} for $\Delta_b \geq 2$ rounds and executes \textsc{{MBRB.Broadcast($m'$)}} in round $r_1$. Subsequently, $p_s$ remains correct for the successive $\Delta_1$ rounds, it gets permanently \textit{faulty} from round $r_{\Delta_b+\Delta_1+1}$ (namely $\forall r_j \in [r_{\Delta_b+\Delta_1+1}, \infty), ~ p_s \in B(r_j)$), and it executes \textsc{{MBRB.Broadcast($m''$)}} in round $r_{\Delta_b+\Delta_1+1}$.
We remark that the failure state of any process may change unexpectedly due to the movement of a Byzantine agent.
Let us consider another local execution history $\mathcal{H}''_s$ of process $p_s$ where the failure state of $p_s$ evolves in the opposite way from $\mathcal{H}'_s$, that is process $p_s$ is \textit{faulty} in rounds $r_j \in [r_1, r_{\Delta_b+\Delta_1}]$ and executes \textsc{{MBRB.Broadcast($m'$)}} in round $r_1$; 
subsequently, $p_s$ is permanently \textit{correct} from round $r_{\Delta_b+\Delta_1+1}$ (namely $\forall r_j \in [r_{\Delta_b+\Delta_1+1}, \infty), p_s \in C(r_j)$) and executes \textsc{{MBRB.Broadcast($m''$)}} in round $r_{\Delta_b+\Delta_1+1}$.
Notice that in both histories $p_s$ executes the \textsc{MBRB.Broadcast} operation only once while correct.
We provide a graphical representation of the two histories in Figure \ref{fig0}.
Let us consider a process $p_1 \neq p_s$ that is correct for the entire lifetime of the system (i.e. $\forall r_j, ~ p_1 \in C(r_j)$), thus $p_1$ is also an $\Delta_c$-infinitely often correct process for any value of $\Delta_c \in \mathbb{N}$.
The two execution histories $\mathcal{H}'_s$ and $\mathcal{H}''_s$ are indistinguishable to $p_1$ because the same operations and events occurred on $p_s$.
Process $p_1$ is not aware of the failure state of $p_s$ (i.e. it has no access to the failure oracle on $p_s$). 
Even defining an algorithm $\mathcal{A}$ that allows process $p_s$ to share the information obtained from $\mathcal{O}_{\mathit{FFA}}$ with process $p_1$ through the point-to-point primitive, process $p_1$ cannot distinguish an execution of $\mathcal{A}$ where $p_s$ is correct and reveals a previous faulty state, from another where $p_s$ is faulty, and maliciously reports the same information.

According to the \textit{Validity} property of the MBRB specification, process $p_1$ executing $\mathcal{P}$ must MBRB-deliver a message from $p_s$ considering both histories because process $p_s$ MBRB-broadcasts a message when correct.
If $\mathcal{P}$ makes process $p_1$ eventually MBRB-deliver message $m'$, then the \textit{Validity} property is violated in $\mathcal{H}''_s$, because process $p_1$ never MBRB-delivers $m''$ (according to the \textit{No-duplication} property) that is broadcast when $p_s$ is correct.
If $\mathcal{P}$ makes process $p_1$ eventually MBRB-deliver message $m''$, then the \textit{Validity} property is violated in $\mathcal{H}'_s$ for the same reason.
This is a contradiction and the claim follows regardless of the value of $\Delta_b$ and $\Delta_c$.
\end{proof}

\begin{figure}
     \centering
     \begin{subfigure}[b]{0.49\textwidth}
         \centering
         \includegraphics[width=\textwidth]{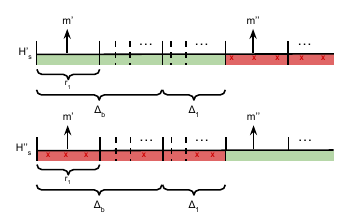}
         \caption{Graphical representations for Theorem \ref{th:mbrbimp1}}
        \label{fig0}
     \end{subfigure}
     \hfill
     \begin{subfigure}[b]{0.49\textwidth}
         \centering
         \includegraphics[width=\textwidth]{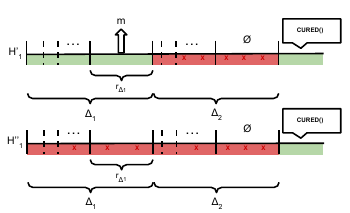}
         \caption{Graphical representations for Theorem \ref{th:mbrbimp2}.}
         \label{fig1}
     \end{subfigure}
    \caption{Graphical representations for Theorems' proof.}
    \label{figX}
\end{figure}

\noindent Theorem \ref{th:mbrbimp1} states the impossibility in solving MBRB assuming the most constrained assumptions we considered. Corollary \ref{cor:1} extends the result to less constrained settings.

\begin{corollary}
    \label{cor:1} If $\Delta_b \in \mathbb{N}^{+}$ and $\Delta_b \geq 2$ rounds, then
there exists no protocol $\mathcal{P}$ implementing a Mobile Byzantine Reliable Broadcast primitive in $\langle$\texttt{SYNC}, \texttt{S-MOB}$^{+},\mathcal{O}_{\mathit{BFA}}
\rangle$ or in $\langle\texttt{SYNC}, \texttt{S-MOB}, * \rangle$.
\end{corollary}

    \begin{proof}
        \let\clearpage\relax 
%
        The claim follows from the same argument provided in Theorem \ref{th:mbrbimp1} give that the considered settings assume either a less constrained agent mobility model (\texttt{S-MOB}) or a failure oracle providing less knowledge ($\mathcal{O}_{\mathit{BFA}}$).
    \end{proof}

\begin{theorem}
    \label{th:mbrbimp2}
    If $\Delta_b \in \mathbb{N}^{+}$ and $\Delta_b \geq 2$ rounds, then
    there exists no protocol $\mathcal{P}$ implementing a Mobile Byzantine Reliable Channel primitive in $\langle \texttt{SYNC}, \texttt{S-MOB}^{+}, \mathcal{O}_{BFA}\rangle$.
\end{theorem}

\begin{proof}
For the sake of contradiction, let us assume that such a protocol $\mathcal{P}$ exists.
Let us assume a permanently correct process $p_s$ (i.e. $\forall r_j, ~ p_s \in C(r_j)$) that executes \textsc{{MBBC.Broadcast($m$)}} in rounds $r_1$.
Let us consider the local execution history $\mathcal{H}'_1$ of a process $p_1$ that is \textit{correct} in rounds $r_j \in [r_1, r_{\Delta_1}], ~ \Delta_1 \in \mathbb{N}$, and executes \textsc{{MBBC.Deliver($m$)}} in round $r_{\Delta_1}$;
subsequently, $p_1$ gets \textit{faulty} for $\Delta_2$ consecutive rounds, $\Delta_2 \in \mathbb{N}$, it wipes its local state (i.e. initialises all the process variables) in round $r_{\Delta_1+\Delta_2}$, and it gets \textit{permanently correct} from round $r_{\Delta_1+\Delta_2+1}$ (namely $\forall r_i \in [r_{\Delta_1+\Delta_2+1}, \infty), ~p_1 \in C(r_i)$).\\
Let us consider another local execution history $\mathcal{H}''_1$ of process $p_1$ that is \textit{faulty} in rounds $r_j \in [r_1, r_{\Delta_1+\Delta_2}]$ and it wipes its local state in round $r_{\Delta_1+\Delta_2}$;
subsequently, $p_1$ gets \textit{permanently correct} from round $r_{\Delta_1+\Delta_2+1}$ (namely $\forall r_j \in [r_{\Delta_1+\Delta_2+1}, \infty), ~p_1 \in C(r_j)$).
We provide a graphical representation in Figure \ref{fig1}.
In round $r_{\Delta_1+\Delta_2+1}$, process $p_1$ has the same local state in both histories and the $\mathcal{O}_{\mathit{BFA}}$ oracle generates the same \textsc{cured}$()$ event on process $p_1$.
Process $p_1$ does not know what happened during the previous rounds.
It is even defining an algorithm $\mathcal{A}$ that allows any process $p_i$ to share and retrieve the state and events occurred on the process through the point-to-point primitive: process $p_i$ can execute such a protocol either as correct or as faulty, and the two executions would be indistinguishable by any other process.

According to the \textit{Validity} property of the MBBC specification, process $p_1$ executing $\mathcal{P}$ must MBBC-deliver message $m$ from $p_s$ in both histories. In round $r_{\Delta_1+\Delta_2+1}$ process $p_1$ has the same local state on both histories, thus it can act in one only way, specifically it can command or not process $p_i$ to deliver message $m$ from $p_s$. In the positive case, the protocol violates the \textit{No duplication} property in history $\mathcal{H}'_1$, in the negative case the \textit{Validity} property is violated by the protocol in $\mathcal{H}''_1$. This leads to a contradiction and the claim follows regardless to the value of $\Delta_1$,$\Delta_2$, and $\Delta_c$.
\end{proof}




{\footnotesize

\begin{table}
    \resizebox{0.85\textwidth}{!}{
    \begin{subtable}[c]{0.49\textwidth}
    \begin{tabular}{c|c|cc|}
    \cline{2-4}
        & ASYNC                & \multicolumn{2}{c|}{SYNC}       \\ \hline
    \multicolumn{1}{|c|}{\multirow{3}{*}{S-MOB$^{+}$}} & \multirow{3}{*}{} & \multicolumn{1}{c|}{$\mathcal{O}_{\mathit{BFA}}$} & $\mathcal{O}_{\mathit{FFA}}$ \\ \cline{3-4} 
    \multicolumn{1}{|c|}{}                             &                      & \multicolumn{1}{c|}{\xmark} & \xmark  \\
    \multicolumn{1}{|c|}{}                             &                 & \multicolumn{1}{c|}{(Cor. \ref{cor:1})} & (Th. \ref{th:mbrbimp1})   \\ \hline
    \multicolumn{1}{|c|}{\multirow{3}{*}{S-MOB}}       & \multirow{2}{*}{\xmark}   & \multicolumn{1}{c|}{$\mathcal{O}_{\mathit{BFA}}$}   & $\mathcal{O}_{\mathit{FFA}}$    \\ \cline{3-4} 
    \multicolumn{1}{|c|}{}                             &                      & \multicolumn{1}{c|}{\xmark} & \xmark  \\
    \multicolumn{1}{|c|}{}                             & (Cor. \ref{cor:2})                  & \multicolumn{1}{c|}{(Cor. \ref{cor:1})} & (Cor. \ref{cor:1})  \\ \hline
\multicolumn{1}{|c|}{\multirow{2}{*}{A-MOB}}       & \xmark                    & \multicolumn{2}{c|}{\xmark}          \\
\multicolumn{1}{|c|}{}                             & (Cor. \ref{cor:2})                  & \multicolumn{2}{c|}{(Th. \ref{th:impAsy2})}         \\ \hline
    \end{tabular}
    \subcaption{MBRB}
    \end{subtable}
    \hspace{30pt}
    \begin{subtable}[c]{0.49\textwidth}
    \begin{tabular}{c|c|cc|}
\cline{2-4}
    & ASYNC                & \multicolumn{2}{c|}{SYNC}       \\ \hline
    \multicolumn{1}{|c|}{\multirow{3}{*}{S-MOB$^{+}$}} & \multirow{3}{*}{} & \multicolumn{1}{c|}{$\mathcal{O}_{\mathit{BFA}}$} & $\mathcal{O}_{\mathit{FFA}}$ \\ \cline{3-4} 
    \multicolumn{1}{|c|}{}                             &                      & \multicolumn{1}{c|}{{\color{purple}\xmark (\textbf{* Sec \ref{sec:weakprim}})}} & {\color{purple}\cmark}  \\
    \multicolumn{1}{|c|}{}                             &                    & \multicolumn{1}{c|}{(Th. \ref{th:mbrbimp2})} & {\color{purple}\textbf{(Th. \ref{th:main}})}  \\ \hline
    \multicolumn{1}{|c|}{\multirow{3}{*}{S-MOB}}       & \multirow{2}{*}{\xmark}   & \multicolumn{1}{c|}{$\mathcal{O}_{\mathit{BFA}}$}   & $\mathcal{O}_{\mathit{FFA}}$    \\ \cline{3-4} 
    \multicolumn{1}{|c|}{}                             &                      & \multicolumn{1}{c|}{} &   \\
    \multicolumn{1}{|c|}{}                             & (Cor. \ref{cor:2})                  & \multicolumn{1}{c|}{\xmark} & \textbf{?}  \\ \hline
    \multicolumn{1}{|c|}{\multirow{2}{*}{A-MOB}}       & \xmark                    & \multicolumn{2}{c|}{\xmark}          \\
    \multicolumn{1}{|c|}{}                             & (Cor. \ref{cor:2})                  & \multicolumn{2}{c|}{(Th. \ref{th:impAsy2})}         \\ \hline
\end{tabular}
    \subcaption{MBBC}
    \end{subtable}
    }
\caption{Summary of the solvability results.}\label{tab:summaryImpossibilities}
\end{table}
}

\noindent {\bf Discussion.} {Contrarily to what we could expect, the MBRB and MBBC problems are impossible to solve in settings (\emph{e.g.}, $\langle$\texttt{SYNC}, \texttt{S-MOB}$^{+}$, $\mathcal{O}_{\mathit{NFA/BFA}}\rangle$) where the register abstraction and consensus problems are solvable~\cite{DBLP:conf/wdag/Garay94,banu2012improved,DBLP:conf/opodis/SasakiYKY13,DBLP:journals/tcs/BonnetDNP16, DBLP:conf/podc/BonomiPPT16, DBLP:conf/icdcn/BonomiPP16, DBLP:conf/srds/BonomiPPT17}. 
{
The intuition behind this is that other problems addressed under the MBF model have a semantics that do not require to execute a particular operation (the delivery of a message in our case) at most once and depending on a precedent failure state of the process. Indeed, both the register abstractions and consensus set constraints on a local value stored by the processes (respectively, the shared value and the decided value) but no primitive operation is associated with their update in their specification. 
Contrarily, MBRB and MBBC introduce constraints on the deliveries of messages that depend on the actual and previous failure states of the processes, generating thus symmetry conditions that are impossible to break without violating one of the properties characterizing the specification. In particular, the main challenge is to ensure that a single broadcast instance does not generate multiple deliveries to the same process while it is correct.
}
%
%
Another counter-intuitive result is that considering a setting stronger than the one considered in related works (\emph{e.g.}, $\langle$\texttt{SYNC}, \texttt{S-MOB}$^{+}$, $\mathcal{O}_{\mathit{FFA}}\rangle$), the MBRB problem is impossible to solve while the MBBC one is possible (see Section~\ref{sec:prot}). 
In the static Byzantine failure model (where every process is always either correct or faulty in a given execution), the channel specification extends the broadcast one allowing multiple broadcast from the same source. As a matter of fact, in the mobile Byzantine failure model such an extension is less constrained with respect to the broadcast: in MBRB, every process can execute only one broadcast operation for the entire lifetime of the system, whereas MBBC allows multiple broadcasts from the same source; if a process is faulty and executes a broadcast, then it is not allowed to execute a subsequent broadcast when correct in the future in the MBRB specification (\textit{No duplication} property), while it is in MBBC. 
Finally, note that other primitives, such as consensus or register abstractions, are not useful in solving the MBRB/MBBC problems. Consider again the execution depicted in Figure \ref{fig0}, correct process may agree or may store a set of delivered messages (according to the MBRB/MBBC specifications) but a single process ($p_s$ in the example), in the settings we characterized, cannot infer if it has already delivered or not a message if it was previously compromised.

\section{A Protocol for MBBC in $\langle$\texttt{SYNC, S-MOB$^{+}$, $\mathcal{O}_{\mathit{FFA}}$}$\rangle$}
\label{sec:prot}
Theorem \ref{th:mbrbimp2} and Corollary \ref{cor:1} motivate the definition of a stronger local oracle than those considered in related work dealing with mobile Byzantine faults, $\mathcal{O}_{FFA}$: both MBRB and MBBC are impossible to solve in the ($\langle$\texttt{SYNC}, \texttt{S-MOB}$^{+}$, $\mathcal{O}_{\mathit{NFA/BFA}}\rangle$) settings.
Theorem \ref{th:mbrbimp1} states the impossibility in solving MBRB even in ($\langle$\texttt{SYNC}, \texttt{S-MOB}$^{+}$, $\mathcal{O}_{\mathit{FFA}}\rangle$). 
This Section investigates the remaining open problem-setting: the solvability of MBBC in ($\langle$\texttt{SYNC}, \texttt{S-MOB}$^{+}$, $\mathcal{O}_{\mathit{FFA}}\rangle$). 
Specifically, we start by defining $\mathcal{P}_{MBBC-RB}$, a protocol implementing the MBBC($\Delta_b, \Delta_c $) communication primitive.
Then, we prove its correctness and fault-tolerance optimality.

\subsection{$\mathcal{P}_{MBBC-RB}$: Protocol Description}
$\mathcal{P}_{MBBC-RB}$ is an extension of Bracha's algorithm \cite{DBLP:journals/iandc/Bracha87} aimed to solve the MBBC problem. 
It inherits Bracha's diffusion mechanism: a payload message $m$ is exchanged inside three protocol messages,  \textsf{SEND}, \textsf{ECHO}, and \textsf{READY}. 
The former is initially sent by the source process to all peers, and the latter are subsequently diffused by all correct processes to all peers if certain conditions are met, namely certain quorums are reached.\\
The pseudo-code of $\mathcal{P}_{MBBC-RB}$ is shown in Algorithm \ref{alg:mbrbrb1}. This solution overcomes the impossibility stated in Theorem \ref{th:mbrbimp2} by leveraging on $\mathcal{O}_{FFA}$ and by fixing the round index (i.e., the moment in time) where MBBC-deliveries must occur.
Every protocol's message contains the information about a specific MBBC-broadcast instance, specifically the source process label $s$, the message (payload) $m$, and the round counter $r_b$ when the broadcast instance started.
An MBBC-broadcast instance proceeds in four consecutive rounds in $\mathcal{P}_{MBBC-RB}$.
In the first round $r_b$, the protocol's message \textsf{SEND} is computed by $p_s$ and enqueued to \textsf{P2P}-send to all processes in the subsequent round.
Every process that \textsf{P2P}-receives a \textsf{SEND} message in round $r_{b+1}$ from $p_s$ computes the \textsf{ECHO} protocol's message for $\langle s, r_b, m \rangle$ and enqueues it to \textsf{P2P}-send to all peers. 
In round $r_{b+2}$, the processes that receive sufficiently many \textsf{ECHO} messages (more than $(n + f)/2$) for an MBBC-broadcast instance from distinct peers generate the related \textsf{READY} protocol's message to \textsf{P2P}-send to all processes. 
Finally, in round $r_{b+3}$, the processes that receive a sufficient number of \textsf{READY} messages (more than $2f$) for an MBBC-broadcast instance from distinct peers MBBC-deliver the associated message $m$ from $p_s$.
An additional protocol's message with respect to Bracha~\cite{DBLP:journals/iandc/Bracha87}, i.e. \textsf{ABORT}, is exchanged in $\mathcal{P}_{MBBC-RB}$ to guarantee the \textit{Agreement} property in case of a faulty source.
In $\mathcal{P}_{MBBC-RB}$, if a correct process $p_s$ executes \textsc{{MBBC.Broadcast($m$)}} in round $r_b$, then every process that is correct in round $r_{b+3}$ triggers \textsc{MBBC.Deliver}($s$,$m$) in the \textit{compute} phase of that round; every process that is faulty in round $r_{b+3}$ MBBC-delivers the message $m$ from $p_s$ at the first round $r_k > r_{b+3}$ it is correct.

We plug the fault-tolerant round counter defined by Bonnet at al. \cite{DBLP:journals/tcs/BonnetDNP16} inside the $\mathcal{P}_{MBBC-RB}$ protocol, enabling all correct processes to share the same value for the round index (that is assumed as an integer value). Its purpose is to fix the single round where the delivery of a certain message can take place. The round counter features are summarised in the following remark.

\begin{remark}[Round counter correctness \cite{DBLP:journals/tcs/BonnetDNP16}]
    \label{roundcountercorrect}
     In $\langle$\texttt{SYNC, S-MOB$^{+}$, $\mathcal{O}_{\mathit{BFA/FFA}}$}$\rangle$, if $n>3f$ then every correct process $p_i$ in round $r_j$ stores the same value for the round index (namely the variable $rc$ in Algorithm \ref{alg:mbrbrb1}) during compute phase.
\end{remark}




We stress the fact that protocol's messages in $\mathcal{P}_{MBBC-RB}$ (\textsf{SEND}, \textsf{ECHO}, \textsf{READY}, and \textsf{ABORT}) must be propagated in specific rounds with respect to the beginning of the MBBC-broadcast, in order to progress till the delivery of the associated message $m$.

{For ease of better understanding, we give a detailed description of $\mathcal{P}_{MBBC-RB}$ in Appendix \ref{sec:variables}, and we illustrate some examples of its execution in Appendix \ref{sec:executions} and within the proof of Lemma \ref{lm:sufficiency}.}

\begin{algorithm}[!h]
    \scriptsize
	\caption{$\mathcal{P}_{MBBC-RB}$}
	\begin{algorithmic}[1] 
	
    \Procedure{Init}{}
        \Let{To\_send}{$\emptyset$, Sends $\leftarrow \emptyset$, cured $\leftarrow$ \textsf{False}, rc $\leftarrow 1$} 
        
        \Let{Echos}{$\{\}$, Readys $\leftarrow \{\}$, Aborts $\leftarrow \{\}$} \Comment{map, $\langle s,r,m\rangle$ : set of process ids}
        \Let{RC}{$\{\}$} \Comment{map, process id : round value}
    \EndProcedure
    
	\Procedure{Broadcast}{m}
	        \Let{To\_send}{To\_send $\cup~\{\langle \textsf{SEND}, \text{s}, rc, m \rangle\}$} \label{algo:prepsend}
	\EndProcedure
    
	\Upon{\textsc{$\mathcal{O}_{\mathit{FFA}}$.cured}} \label{alg:fail1}
	    \Let{cured} \label{alg:fail2}{\textsf{True}} \label{algo:setcured}
	\EndUpon
	
	\Statex \textit{\textbf{Send Phase}}
	\If{cured} \label{algo:ifcured}
	    \Let{To\_send}{$\emptyset$} \label{algo:wipetosend}
	\EndIf
	\For{$\text{pk} \in \text{To\_send}$ } \label{algo:lineb1}
	    \For{$q \in \Pi$}
    	    \State ${\sf P2P.send}(q,\text{pk})$   \label{algo:lineb2}
	    \EndFor
	\EndFor
	
	\Statex {\textit{\textbf{Receive Phase}}}
    \Let{Sends}{$\emptyset$, Echos $\leftarrow \{\}$, Readys $\leftarrow \{\}$, Aborts $\leftarrow \{\}$, RC $\leftarrow \{\}$}
    \Upon{${\sf P2P.deliver}(\text{q}, \langle \textsf{Type}, \text{s}, r_b, m \rangle)$} \label{algo:linep2del}
        \If{$s = q$ \textbf{and} \textsf{Type} = \textsf{SEND}} \label{algo:linep2delsend1}
            \Let{Sends}{Sends $\cup ~\{\langle s,r_b, m \rangle\}$} \label{algo:linep2delsend2}
        \EndIf
        \If{\textsf{Type} = \textsf{ECHO}} \label{algo:linep2delecho1}
            \Let{Echos[$\langle s,r_b,m \rangle$]}{Echos[$\langle s,r_b,m \rangle$] $\cup ~\{q\}$} \label{algo:linep2delecho2}
        \EndIf
        \If{\textsf{Type} = \textsf{READY}}
            \Let{Readys[$\langle s,r_b, m \rangle$]}{Readys[$\langle s,r_b,m \rangle$] $\cup ~\{q\}$}
        \EndIf
        \If{\textsf{Type} = \textsf{ABORT}}
            \Let{Aborts[$\langle s,r_b,m \rangle$]}{Aborts[$\langle s,r_b,m \rangle$] $\cup ~\{q\}$}
        \EndIf
	\EndUpon
    \Upon{${\sf P2P.deliver}(\text{q}, \langle  \textsf{ROUND}, \text{j} \rangle)$}
        \Let{RC[$q$]}{$j$} \label{mbrbrblastline}
	\EndUpon
	
		\Statex \textit{\textbf{Compute Phase}}
    	\Let{To\_send}{$\emptyset$, rc $\leftarrow$ \textsc{getMajority}(RC.\textsc{values})}
    	\For{$\langle s,r_b,m \rangle \in \text{Sends}$} \label{algo:computesendbegin}
    	    \If{$\mathrm{rc} = r_{b+1}$} \label{algo:ifintegrity2}
    	        \Let{To\_send}{To\_send $\cup ~\{\langle \textsf{ECHO}, \text{s}, r_b, \text{m} \rangle\}$ } \label{algo:computesendend}
    	    \EndIf
    	\EndFor
    	\For{$\langle s,r_b,m \rangle \in \text{Echos}$}
    	\If{$|\mathrm{Echos}[\langle s,r_b,m \rangle]| > (n + f)/2$ } \label{algo:computeechobegin}
    	    \Let{To\_send}{To\_send $\cup ~\{\langle \textsf{READY}, \text{s}, r_b, \text{m} \rangle\}$ } \label{algo:computeechoend}
		\ElsIf{$|\mathrm{Echos}[\langle s,r_b,m \rangle]| > f$} \label{algo:elseif}
		    \Let{To\_send}{To\_send $\cup ~\{\langle \textsf{ABORT}, \text{s}, r_b, \text{m} \rangle\}$ }
        \EndIf
    	\EndFor
	   	\For{$\langle s,r_b,m \rangle \in \text{Aborts}$}
    	\If{$|\mathrm{Aborts}[\langle s,r_b,m \rangle]| > f$ } \label{algo:ifabort}
    	    \Let{$|\mathrm{Readys}[\langle s,r_b,m \rangle]$}{$\emptyset$}
        \EndIf
    	\EndFor
    	\For{$\langle s,r_b,m \rangle \in \text{Readys}$}
    	\If{$|\mathrm{Readys}[\langle s,r_b,m \rangle]| > 2f$} \label{algo:computedeliveryif1}
    	    \If{(($\mathrm{rc} = r_{b+3}$) \textbf{or} (cured \textbf{and} $\mathrm{rc} > r_{b+3}$ \textbf{and} \textsc{$\mathcal{O}_{\mathit{FFA}}$.faulty\_at} $\leq r_{b+3}$)) \Statex \hspace{30pt} \textbf{and} ($\nexists \langle s, r_k,m \rangle \in \mathrm{Readys} : (|\mathrm{Readys} \langle s, r_k,m \rangle| > 2f) \wedge (r_k < r_b)$)} \label{algo:computedeliveryif2}
    	        \State \textsc{Deliver}(s,m) \label{algo:computedelivery}
    	    \EndIf
    	    \Let{To\_send}{To\_send $\cup ~\{\langle \textsf{READY}, \text{s}, r_b, \text{m} \rangle\}$} \label{algo:continueready}
        \EndIf
    	\EndFor
    	\Let{cured}{\textsf{False}, rc $\leftarrow$ rc+1, To\_send $\leftarrow$ To\_send $\cup ~\{\langle \textsf{ROUND}, \text{rc} \rangle\}$}
		
	\end{algorithmic}
 	\label{alg:mbrbrb1}
\end{algorithm}

\subsection{Correctness Proofs}
We remark that in \texttt{S-MOB$^{+}$} mobile agents can move only between the \textit{compute} and \textit{send} phase of two consecutive rounds. This implies that $\Delta_s$ is assumed greater than or equal to one round. 
Such mobility model has the following effects to the agents' capabilities: at the beginning of a round $r_j$, mobile agents can potentially control the messages that are diffused by $2f$ processes, the ones where the mobile agents are placed in $r_j$ and the others where they were in the previous round $r_{j-1}$ (they can set in round $r_{j-1}$ the messages that will be exchange by freed processes in round $r_j$).
This capability can partially be mitigated by the local failure detector $\mathcal{O}_{\mathit{FFA}}$: a process can discard all messages queued to be send right after the failure detector notifies the \textsc{cured}$()$ event. It follows that, at the beginning of a round, at most $f$ processes may not participate in the protocol and at most $f$ may have a Byzantine behavior.

The following Lemmas and Theorem state the correctness of $\mathcal{P}_{MBBC-RB}$ in solving the MBBC problem and its fault-tolerance optimality with respect to the number of tolerated mobile agents.	

\begin{lemma}
    \label{lm:sufficiency}
    If $\Delta_b \geq 2$ rounds and $\Delta_c \geq 1$ round, then $\mathcal{P}_{MBBC-RB}$ solves the Mobile Byzantine Broadcast Channel problem (MBBC) in $\langle$\texttt{SYNC, S-MOB$^{+}$, $\mathcal{O}_{\mathit{FFA}}$}$\rangle$ if $n > 5f$.
\end{lemma}

\begin{proof}
For simplicity, we give the proof assuming the minimum values for $\Delta_b$ and $\Delta_c$. The arguments extend to higher values.

\underline{\textit{$\left(\Delta_b =2 \text{ rounds}, ~\Delta_c=1 \text{ round}\right)$-Validity}}: We prove that if we assume $\Delta_b=2$ rounds, $\Delta_c=1$ round, and a process $p_s$ is correct in round $r_b$ when it executes \textsc{MBBC.Broadcast}($m$), then every process that is $\Delta_c$-infinitely often correct eventually triggers \textsc{MBBC.Deliver}($s$,$m$), that implies the \textit{$\left(\Delta_b,\Delta_c\right)$-Validity} property.
    The MBBC-delivery of a message $m$ from a process $p_s$ may occur either because $p_s$ was correct in round $r_b$ and executed \textsc{MBBC.Broadcast}($m$) or since $p_s$ was faulty at some round $r_d< r_b$ and ${\sf P2P}$-sent a SEND message with payload $m$.
    Let us assume that process $p_s$ has not ${\sf P2P}$-sent yet the SEND message with payload $m$ neither as correct or faulty before round $r_b$, that it is correct in rounds $r_b$ and $r_{b+1}$ ($\Delta_b = 2$) and executes the procedure \textsf{Broadcast} with parameter $m$ in round $r_b$. 
    The $\langle \textsf{SEND}, \text{s}, r_b, \text{m} \rangle$ message is then prepared (line \ref{algo:prepsend}) to be relayed to all other processes (lines \ref{algo:lineb1}-\ref{algo:lineb2}).
    In round $r_{b+1}$, the $\langle \textsf{SEND}, \text{s}, r_b, \text{m} \rangle$ message  is \textsf{P2P}-sent by $p_s$ to all processes and it is received by all but $f$ (the ones controlled by mobile agents); it follows that $n-f$ processes executes lines \ref{algo:linep2del}-\ref{algo:linep2delsend2} during the \textit{receive} phase in round $r_{b+1}$ and lines \ref{algo:computesendbegin}-\ref{algo:computesendend} in the \textit{compute} phase, preparing the $\langle \textsf{ECHO}, \text{s}, r_b, \text{m} \rangle$ message to \textsf{P2P-send} in round $r_{b+2}$.
    In round $r_{b+2}$, at least $n-2f$ processes relay the message $\langle \textsf{ECHO}, \text{s}, r_b, \text{m} \rangle$ ($f$ process may be faulty in round $r_{b+2}$ and $f$ process may have been faulty in round $r_{b+1}$) and it is received by $n-f$ processes (again, the ones not controlled by mobile agents). These processes execute lines \ref{algo:linep2del}, \ref{algo:linep2delecho1} and \ref{algo:linep2delecho2} in the \textit{receive} phase and lines \ref{algo:computeechobegin} and \ref{algo:computeechoend} in the \textit{compute} phase. 
    In particular, the condition inside the \textit{if} statement at line \ref{algo:computeechobegin} is verified due to the assumption $n > 5f$, given that $n - 2f > (n+f)/2$, and line \ref{algo:computeechoend} is executed preparing $\langle \textsf{READY}, \text{s}, r_b, \text{m} \rangle$ message to ${\sf P2P}$-send in round $r_{b+3}$.
    Finally, in round $r_{b+3}$, the same reasoning given for round $r_{b+2}$ applies and $n-f$ processes execute lines \ref{algo:computedeliveryif1}-\ref{algo:computedelivery}, given $n - 2f > 2f$ and Remark \ref{roundcountercorrect}, and thus they trigger \textsf{Deliver} with parameters $s$ and $m$.
    At every round $r_j > r_{b+3}$ the $\langle \textsf{READY}, \text{s}, r_b, \text{m} \rangle$ message is ${\sf P2P}$-sent by all the correct processes not faulty in round $r_{j-1}$ (that are at least $n-f$). The \textit{if} statement at line \ref{algo:computedeliveryif2} guarantees that every process that was faulty in round $r_{b+3}$ delivers message $m$ from $p_s$ at the first round $r_k > r_{b+3}$ it is correct.
    Finally, in case \emph{(i)} process $p_s$ was faulty and ${\sf P2P}$-sent the SEND message with payload $m$ in round $r_k < r_b$, \emph{(ii)} every $\Delta_c$-infinitely correct process MBBC-delivered $m$ from $p_s$, and  \emph{(iii)} process $p_s$ is correct in round $r_b > r_k$ and executes \textsc{MBBC.Broadcast}($m$), then the claim still follows: the message $m$ has been already MBBC-delivered (further details can be found in the \textit{Agreement} property's proof).
    
    \underline{\textit{No duplication}}:
    The second sub-condition of the \textit{if} statement at line \ref{algo:computedeliveryif2} guarantees that the entire \textit{if} statement is verified only for the minimum $r_j$ among all the tuples $\langle s, *, m \rangle$ (i.e. the MBBC-delivery is independent from the $r_b$ parameter).
    The first sub-condition inside the \textit{if} statement at line \ref{algo:computedeliveryif2} is verified only once among all the rounds a mobile agent does not control the process.
    More in detail, if the $cured$ variable is $\textsc{False}$, the condition is verified only in round $r_{b+3}$ for the tuple $\langle s, r_b, m \rangle$.
    Otherwise, the \textit{if} statement in line \ref{algo:computedeliveryif2} is verified in round $r_k > r_{b+3}$ when a mobile agent, arrived on the process in round $r_j \leq r_{b+3}$, leaves the process, that occurs only once on a process during the entire lifetime of the system given Remark \ref{roundcountercorrect}.
    The condition $rc>r_{b+3}$ in line \ref{algo:computedeliveryif2} is not required but simplifies this proof.  

    \underline{\textit{$\left(\Delta_b=2\right)$-Integrity}}:  
    For the sake of contradiction, let us assume that a process $p_i$ is correct in round $r_k$ and executes  \textsc{MBBC.Deliver}($s$,$m$), that process $p_s$ is correct in rounds $r_b$ and $r_{b+1}$ (that is, $\Delta_b = 2$), and that it does not execute \textsc{MBBC.Broadcast}($m$) in round $r_b$.\\
    Process $p_i$ MBBC-delivers $m$ from $p_s$ either in round $r_k = r_{b+3}$ if $p_i$ is correct, or at the first round $r_k > r_{b+3}$ when $p_i$ is correct.
    In the former case, more than $2f$ processes sent message $\langle \textsf{READY}, \text{s}, r_b, \text{m} \rangle$ in round $r_{b+3}$, therefore more than $(n+f)/2$ processes sent message $\langle \textsf{ECHO}, \text{s}, r_b, \text{m} \rangle$ in round $r_{b+2}$, that implies that at least $(n+f)/2 -f$ processes were correct in round $r_{b+1}$ and received $\langle \textsf{SEND}, \text{s}, r_b, \text{m} \rangle$ in round $r_{b+1}$ from $p_s$ (lines \ref{algo:ifintegrity2}-\ref{algo:computesendend}).
    No procedure in $\mathcal{P}_{MBBC-RB}$ allows a correct process $p_s$ to P2P-send $\langle \textsf{SEND}, \text{s}, r_b, \text{m} \rangle$ messages except \textsc{Broadcast}($m$). 
    It follows that the latter scenario occurred and process $p_i$ was faulty in round $r_{b+3}$. As a matter of fact, correct process $p_i$ P2P-received more than $2f$ $\langle \textsf{READY}, \text{s}, r_b, \text{m} \rangle$ messages from distinct processes in round $r_{k}$. 
    For the same reasoning as in the former case, this implies that a correct process $p_s$ sent $\langle \textsf{SEND}, \text{s}, r_b, \text{m} \rangle$ messages but no procedure except \textsc{Broadcast}($m$) allows it.
    This leads to a contradiction and the claim follows.
    
    \underline{\textit{$\left(\Delta_c=1\right)$-Agreement}}:
    We proved, in the \textit{Validity} proof, that this property is satisfied in the case of a correct source.
    Faulty processes cannot collude to make one of the \textit{if} statements at lines \ref{algo:computeechobegin}, \ref{algo:elseif}, \ref{algo:ifabort} and \ref{algo:computedeliveryif1} verified for a message $m$ never sent over the ${\sf P2P}$ links of a process $p_s$. More in detail, the attacker cannot attempt to make any correct process MBBC-deliver a message $m$ from $p_s$ without compromising $p_s$.
    We prove that if $p_s$ is faulty and \textsf{P2P}-sends $\langle \textsf{SEND}, \text{s}, r_b, \text{m} \rangle$ messages in round $r_b$, then either all $\Delta_c$-infinitely correct processes delivers $m$ from $p_s$ or no $\Delta_c$-infinitely correct processes delivers $m$ from $p_s$.
    For the sake of contradiction, let us assume that all $\Delta_c$-infinitely often correct processes but some, $p_1, p_2, \dots, p_i$, MBBC-delivered a message $m$ from $p_s$. It follows that there is no round $r_j$ where more than $2f$ correct processes concurrently ${\sf P2P}$-send $\langle \textsf{READY}, \text{s}, r_b, \text{m} \rangle$.
    This implies that the correct processes that delivered $m$ are at most $2f$. According to the protocol, such processes receive a quorum of \textsf{ECHO} messages and at most $f$ \textsf{ABORT} messages about $m$, to generate the required \textsf{READY} messages. More in detail, they received \textsf{ECHO} messages from at least $2f+1$ correct processes. At that point, the faulty processes decided which correct processes reached the quorum of \textsf{ECHO} messages. Nevertheless, each correct process that did not reach the quorum generated an \textsf{ABORT} message. It follows that at most $f$ correct processes did not reach the quorum, whereas $n - f - f$ processes were correct and generated the \textsf{READY} message, which was disseminated by at least $n - 3f$ of them in the subsequent round. Given that $n > 5f$, at least $2f+1$ correct processes concurrently disseminate a \textsf{READY} message and thus all correct processes in round $r_{b+3}$ must MBBC-deliver it. This lead to a contradiction and the claim follows.
\end{proof}

\begin{lemma}
    \label{lm:necessity}
    The Mobile Byzantine Broadcast Channel problem (MBBC) is solvable in $\langle$\texttt{SYNC, S-MOB$^{+}$, $\mathcal{O}_{\mathit{FFA}}$}$\rangle$ only if $n > 5f$.
\end{lemma}

\begin{proof}
    The claim follows by extending the results proven by Backes and Cachin~\cite{DBLP:conf/dsn/BackesC03} and by Raynal~\cite{DBLP:books/sp/Raynal18}. The former states that the BRB problem can be solved in a static distributed system where at most $t$ processes may fail-stop, and at most $f$ processes are Byzantine, if and only if $n > 3f + 2t$. Similarly, Raynal proved that the BRB problem can be solved in a static distributed system, where $t_l$ processes may not send messages, and $t_s$ processes may send spurious messages (processes may exhibit both behaviors during the lifetime of the system), if and only if $n > 2t_l + t_s$.
    
    Both scenarios can be simulated by an attacker in our system: the mobile agents can continuously alternate between two disjoint sets $P_1$ and $P_2$ of $f$ processes, namely it can turn faulty all processes in $P_1$ in all rounds $r_j, j \in \mathbb{N}$, and all processes in $P_2$ in all rounds $r_{j+1}$, sending spurious messages from process in $P_1$ and no message from peers in $P_2$. Therefore, all processes in $P_1$ send spurious messages (behaving like $f$ Byzantine faulty processes), and all the processes in $P_2$ send no message (like $f$ fail-stop faulty processes), and the claim follows. 
\end{proof}

\begin{theorem}
    \label{th:main}
    The Mobile Byzantine Broadcast Channel problem (MBBC) is solvable in $\langle$\texttt{SYNC, S-MOB$^{+}$, $\mathcal{O}_{\mathit{FFA}}$}$\rangle$ with $\mathcal{O}_{\mathit{FFA}}$ if and only if $n > 5f$.
\end{theorem}

\begin{proof}
    It follows from Lemmas \ref{lm:sufficiency} and \ref{lm:necessity}.
\end{proof}

The following Corollary extends the optimality of $\mathcal{P}_{MBBC-RB}$ to the case of slower agents. In other words, even if the mobile agents are slower we are not able to tolerate more agents solving MBBC.

\begin{corollary}
    \label{cor:slower}
    The Mobile Byzantine Broadcast Channel problem (MBBC) is solvable in $\langle$\texttt{SYNC, S-MOB$^{+}$, $\mathcal{O}_{\mathit{FFA}}$}$\rangle$ if and only if $n > 5f$, for each $\Delta_s\geq 1$ round. Furthermore, the actual value of $\Delta_s$ can be unknown to the processes.
\end{corollary}

   \begin{proof}
    \let\clearpage\relax 
    %
%
    In the $\texttt{S-MOB}^{+}$ model, $\Delta_s$ is expressed as a (strictly positive) number of rounds.
    The claim follow from the fact that whatever number of rounds is specified by $\Delta_s$, all the mobile agents can move in one of the three protocol phases when the SEND, ECHO, or READY messages are exchanged for a broadcast instance.

    Furthermore, the actual value of $\Delta_s$ is irrelevant solving the MBBC problem in (\texttt{SYNC, S-MOB$^{+}$, $\mathcal{O}_{\mathit{FFA}}$}): mobile agent are constrained to move only between two consecutive rounds and the $\mathcal{P}_{MBBC-RB}$ protocol is correct assuming the minimum value for $\Delta_s$ in \texttt{S-MOB$^{+}$} (that is, one round).
    \end{proof}

Note that MBBC and MBBR specifications do not allow processes to be \textit{terminate}, namely to eventually stop propagating messages through the P2P primitive. Intuitively,  processes need to continuously relay the messages in order to enforce \textit{$\Delta_c$-Totality/Agreement} 
and thus allow every temporarily faulty process to eventually deliver a broadcast message. Furthermore, as argued in Section \ref{sec:imp}, processes are not able to infer if a specific process has delivered a message, and thus conclude if all processes delivered a message when correct.
Additional assumptions enabling termination can be considered, such as an upper-bound on the time a process becomes correct when faulty.

\section{MBBC with multiple deliveries}
\label{sec:weakprim}
The impossibilities identified in Section \ref{sec:imp} arise for the general specification we defined. In fact, alternative or weaker specifications could be implementable under weaker assumptions.
    More in detail, we proved that no protocol can solve the MBBC in $\langle$\texttt{SYNC, S-MOB$^{+}$, $\mathcal{O}_{\mathit{BFA}}$}$\rangle$. 
    We therefore investigate the possibility of a weaker primitive that can be realized when the stringent conditions identified in Theorem \ref{th:main} are not satisfied.

    We start by considering the case where no local failure detector is available, that is, the case of $\mathcal{O}_{\mathit{NFA}}$.
    The following Theorem show that a weaker MBBC primitive, where the \textit{No duplication} property is not satisfied, is realizable in $\langle$\texttt{SYNC, S-MOB$^{+}$, $\mathcal{O}_{\mathit{NFA}}$}$\rangle$.

    \begin{theorem}
        \label{th:tomoveout1}
        A weaker Mobile Byzantine Broadcast Channel primitive, not guaranteeing the \emph{No duplication} property, is realizable in $\langle$\texttt{SYNC, S-MOB$^{+}$, $\mathcal{O}_{\mathit{NFA}}$}$\rangle$ if $\Delta_b =2$ rounds, $\Delta_c = 1$ round, and $n > 6f$.
    \end{theorem}

   \begin{proof}
    \let\clearpage\relax 
%
        Let us consider the $\mathcal{P}_{MBBC-RB}$ protocol defined in Algorithm \ref{alg:mbrbrb1}. Let us ignore the lines that interacts with the local failure detector, namely \ref{alg:fail1}, \ref{alg:fail2} and \ref{algo:computedeliveryif2}.
        Let us substitute all the occurrences of parameter $f$ with $\bar{f} = 2f$ in Algorithm \ref{alg:mbrbrb1}.

        The difference with respect the setting considered in Lemma \ref{lm:sufficiency} is that processes are not aware of being compromised. In particular, they may diffuse messages with P2P-links previously generated by mobile agents. As a matter of fact, the protocol is restored right after the mobile agent left the process.

        The proof follows from the same reasoning stated in Lemma \ref{lm:sufficiency} except for \textit{No duplication} considering $\bar{f}$ instead of $f$ in Algorithm \ref{alg:mbrbrb1}.
    \end{proof}

    The following theorem show that having a slightly better oracle about failures, namely $\mathcal{O}_{\mathit{BFA}}$, permits to withstand more Byzantine agents, for the same weaker problem that does \emph{not} guarantees no duplication.
    
    \begin{theorem}
        \label{th:weakbfa}
        A weaker Mobile Byzantine Broadcast Channel primitive, not guaranteeing the \emph{No duplication} property, is realizable in $\langle$\texttt{SYNC, S-MOB$^{+}$, $\mathcal{O}_{\mathit{BFA}}$}$\rangle$ if $\Delta_b =2$ rounds, $\Delta_c = 1$ round, and $n > 5f$.
    \end{theorem}

   \begin{proof}
    \let\clearpage\relax 
        The $\mathcal{O}_{\mathit{BFA}}$ failure oracle enables a correct process just freed from a mobile agent to take corrective actions, specifically to discard all messages queued to be sent in the current round.
        As a matter of fact, the $\mathcal{O}_{\mathit{BFA}}$ oracle does not allow a process to know whether it was correct in a defined period in the past, therefore the same technicality employed in $\mathcal{P}_{MBBC-RB}$ and detailed in Lemma \ref{lm:sufficiency} in the \textit{No duplication} part cannot be adopted. 
        The claim follows combining the argumentation provided in Theorem \ref{th:mbrbimp2} and Lemma \ref{lm:sufficiency}.
    \end{proof}

Abandoning the \textit{No duplication} guarantee, the number of message delivered becomes unbounded: the following theorem shows that it is not possible to bound the number of duplicate messages that are delivered, even assuming an intermediate oracle, namely $\mathcal{O}_{\mathit{BFA}}$.

    \begin{theorem}
        \label{th:infdel}
        Given a constant $\bar{k} \in \mathbb{N}^{+}$, it is not possible to define a weaker Mobile Byzantine Broadcast Channel primitive, not guaranteeing the \emph{No duplication} property, in $\langle$\texttt{SYNC, S-MOB$^{+}$, $\mathcal{O}_{\mathit{BFA}}$}$\rangle$ where
        a message $m$ MBBC-Broadcast by a process $p_s$ is MBBC-Delivered by a process $p_i$ at most $\bar{k}$ time when correct.
    \end{theorem}

   \begin{proof}
    \let\clearpage\relax 
%
        The proof follows by extending the argument provided in Theorem \ref{th:mbrbimp2}. In the defined local execution histories $\mathcal{H}'_1$ and $\mathcal{H}''_1$, it is not possible to define an MBBC primitive where both \emph{No duplication} and \emph{Validity} properties are satisfied for a message $m$ MBBC-Broadcast by a process $p_s$. 
        As a matter of fact, if the \emph{No duplication} has not to be satisfied, process $p_i$ can always deliver message $m$ after the \textsc{cured}$()$ event generated by $\mathcal{O}_{\mathit{BFA}}$.

        Let us extend the execution history $\mathcal{H}''_1$. At round $r_{\Delta_1+\Delta_2+1}$ process $p_i$ executes \textsc{{MBBC.Deliver($m$)}}. Subsequently, the pattern of $\mathcal{H}''_1$ repeats: process $p_i$ get faulty and subsequently correct. Process $p_i$, again, is not able to know whether message $m$ from $p_s$ has been previously MBBC-Delivered, thus it executes \textsc{{MBBC.Deliver($m$)}} to satisfy the \emph{Validity} property.

        It follows that process $p_i$ MBBC-Deliver message $m$ from $p_s$ every time that a mobile agent moves away from $p_i$ with the described procedure. Therefore, if the a mobile agent arrives and frees process $p_i$ $\bar{k}+1$ times after the MBBC-Broadcast, process $p_i$ MBBC-Deliver $\bar{k}+1$ times message $m$ from $p_s$. Alternatively, if process $p_i$ does not MBBC-Deliver $m$ from $p_s$ when it get correct, it may not satisfy the \emph{Validity} property, and the claim follows. 
    \end{proof}

    \begin{corollary}
        \label{cor:ktimesdeliver}
        Suppose a solution to a weaker Mobile Byzantine Broadcast Channel primitive, not guaranteeing the \emph{No duplication} property, in $\langle$\texttt{SYNC, S-MOB$^{+}$, $\mathcal{O}_{\mathit{BFA}}$}$\rangle$. If a process $p_i$ gets faulty and correct $k$ times after the MBBC-Broadcast of a message $m$ from $p_s$, then $p_i$ MBBC-Delivers $m$ from $p_s$ at least $k$ times.
    \end{corollary}

   \begin{proof}
    \let\clearpage\relax 
%
        It follows from the same argument provided for Theorem \ref{th:infdel}. Every time a process $p_i$ is freed by a mobile agent after a MBBC-Broadcast, the process has to decide whether to MBBC-Deliver or not a message $m$ MBBC-Brodcast by a process $p_s$. As a matter of fact, process $p_i$ does not known how many times it has been correct in the past, it is only aware that it has been freed by a mobile agent. It follows that if $p_i$ decide to not MBBC-Deliver once a message $m$ from $p_s$ it may invalidates the \emph{Validity} property and the claim follows.
    \end{proof}

    \begin{theorem}
        \label{th:last}
        Suppose a solution to a weaker Mobile Byzantine Broadcast Channel primitive, not guaranteeing the \emph{No duplication} property, in $\langle$\texttt{SYNC, S-MOB$^{+}$, $\mathcal{O}_{\mathit{NFA}}$}$\rangle$. 
        If a process $p_s$ MBBC-Broadcast a message $m$, then every process $p_i$ must MBBC-Deliver $m$ from $p_s$ infinitely often.
    \end{theorem}

   \begin{proof}
        \let\clearpage\relax 
%
        A correct process $p_i$ at round $r_k$ is not aware of its failure state at all round $r_j, j \in \mathbb{N}, j < k$. It follows that, if $p_i$ does not MBBC-Deliver $m$ from $p_s$, then it may not satisfy the \emph{Validity} property. The argument hold for every round $r_h, h \in \mathbb{N}, h > k$ and the claim follows.
    \end{proof}

\section{Conclusion}
We provided a specification for the Byzantine Reliable Broadcast and Byzantine Broadcast Channel problems in distributed systems affected by mobile Byzantine faults. 
We identified some impossibilities; in particular, 
we showed that both speed constraints on the mobile agents and timing assumptions on the system evolution are required to solve the problems under investigation, and we proved that the Byzantine Reliable Broadcast cannot be solved even in one of the most constrained mobile Byzantine failure models presented so far.
The Byzantine Broadcast Channel problem proved to be solvable, assuming a stronger local failure detector than the ones previously considered in the literature.
{Lastly, we investigated a weaker Byzantine Broadcast Channel primitive, not guaranteeing the \emph{No duplication} property, in settings equivalent to the ones assumed in related works.}
Our results characterise the solvability of a fundamental problem in a general dynamic process failure model, and open the path for research on additional important tasks. 
In particular, to understand the gap that exists between the theoretical model (assumed in this and in related work \cite{banu2012improved,DBLP:journals/tcs/BonnetDNP16, DBLP:conf/wdag/Garay94,DBLP:conf/podc/OstrovskyY91, DBLP:conf/opodis/SasakiYKY13, DBLP:journals/iandc/Reischuk85}) and the practical world, investigating the feasibility of the oracles and defining solutions that are as practical as possible. 
Furthermore, it may be interesting to relax the assumptions of instantaneous fault detection and recovery (of the protocol), to investigate whether the assumption of digitally signed messages has an impact on the solvability of the considered problems, and to analyse the Mobile Byzantine Channel problem assuming the \texttt{S-MOB} agent mobility model (which we have left open for analysis and we conjecture its solvability).

\begin{appendices}

 \section{The Byzantine Reliable Broadcast and Channel Problems Specification \cite{DBLP:journals/iandc/Bracha87, DBLP:books/daglib/0025983}}
\label{sec:brb}
The Byzantine Reliable Broadcast and the Byzantine Broadcast Channel problems aim at specifying a communication primitive, respectively \textsc{BRB} and \textsc{BBC}, exposing two operations, \textsc{BRB/BBC-broadcast($m$)} and \textsc{BRB/BBC-deliver($s,m$)}, where $m$ is a message and $s$ is a process identifier.

The BRB primitive enables all correct processes of a distributed system to agree on a single message diffused by a (potentially faulty) particular process, the source. The BBC primitive extends BRB allowing all processes to diffuse an arbitrary number of messages so that all correct processes eventually deliver the same set of messages.
We say that a process $p_i$ \quotes{BRB/BBC-broadcasts a message $m$} when it invokes \textsc{BRB/BBC-broadcast($m$)}, and $p_i$ \quotes{BRB/BBC-delivers a message $m$ from $p_s$} when it manage the \textsc{BRB/BBC-deliver($s,m$)} event.

We remark that both BRB and BBC primitives assume a \textit{static process failure model} where every process is permanently correct or faulty.

\subsection{Byzantine Reliable Broadcast (BRB)}
The \textsc{BRB} communication primitive guarantees the following properties:
\begin{itemize}
    \item \textit{Validity}:  If a correct process $p_s$ BRB-broadcasts a message $m$, then every correct process eventually BRB-delivers $m$ from $p_s$.
    \item \textit{No duplication}: Every correct process BRB-delivers at most one message from $p_s$.
    \item \textit{Integrity}: If some correct process BRB-delivers a message $m$ from $p_s$ and process $p_s$ is correct, then $m$ was previously BRB-broadcast by $p_s$.
    \item \textit{Consistency}: If some correct process BRB-delivers a message $m$ from $p_s$ and another correct process BRB-delivers a message $m'$ from $p_s$, then $m = m'$.
    \item \textit{Totality}: If some message is BRB-delivered by any correct process, every correct process eventually BRB-delivers a message.
\end{itemize}

\subsection{Byzantine Broadcast Channel (BBC)}
\label{sec:bbc}
%
The \textsc{BBC} communication primitive guarantees the following properties:
\begin{itemize}
    \item \textit{Validity}:  If a correct process $p_s$ BBC-broadcasts a message $m$, then every correct process eventually BBB-delivers $m$ from $p_s$.
    \item \textit{No duplication}:
    No correct process BBC-delivers a message $m$ from $p_s$ more than once.
    \item \textit{Integrity}: If some correct process BBC-delivers a message $m$ from $p_s$ and process $p_s$ is correct, then $m$ was previously BBC-broadcast by $p_s$.
    \item \textit{Agreement}: If some correct process BBC-delivers a message $m$ from $p_s$ then every correct process eventually delivers message $m$ from $p_s$.
\end{itemize}

\section{$\mathcal{P}_{MBBC-RB}$ execution examples}
\label{sec:executions}
\let\clearpage\relax
    We detail in this Section several execution examples for the $\mathcal{P}_{MBBC-RB}$ protocol defined in Section \ref{sec:prot}.  Given what claimed in Theorem \ref{th:main}, we assume that the correctness conditions for our protocol, i.e. a $\langle$\texttt{SYNC, S-MOB$^{+}$, $\mathcal{O}_{\mathit{FFA}}$}$\rangle$ system and $n>5f$, are satisfied in all of the provided examples.
    We detail one example where the source is correct and two in which the source is faulty.

\begin{figure}[h]
\centering
\includegraphics[width=.8\textwidth]{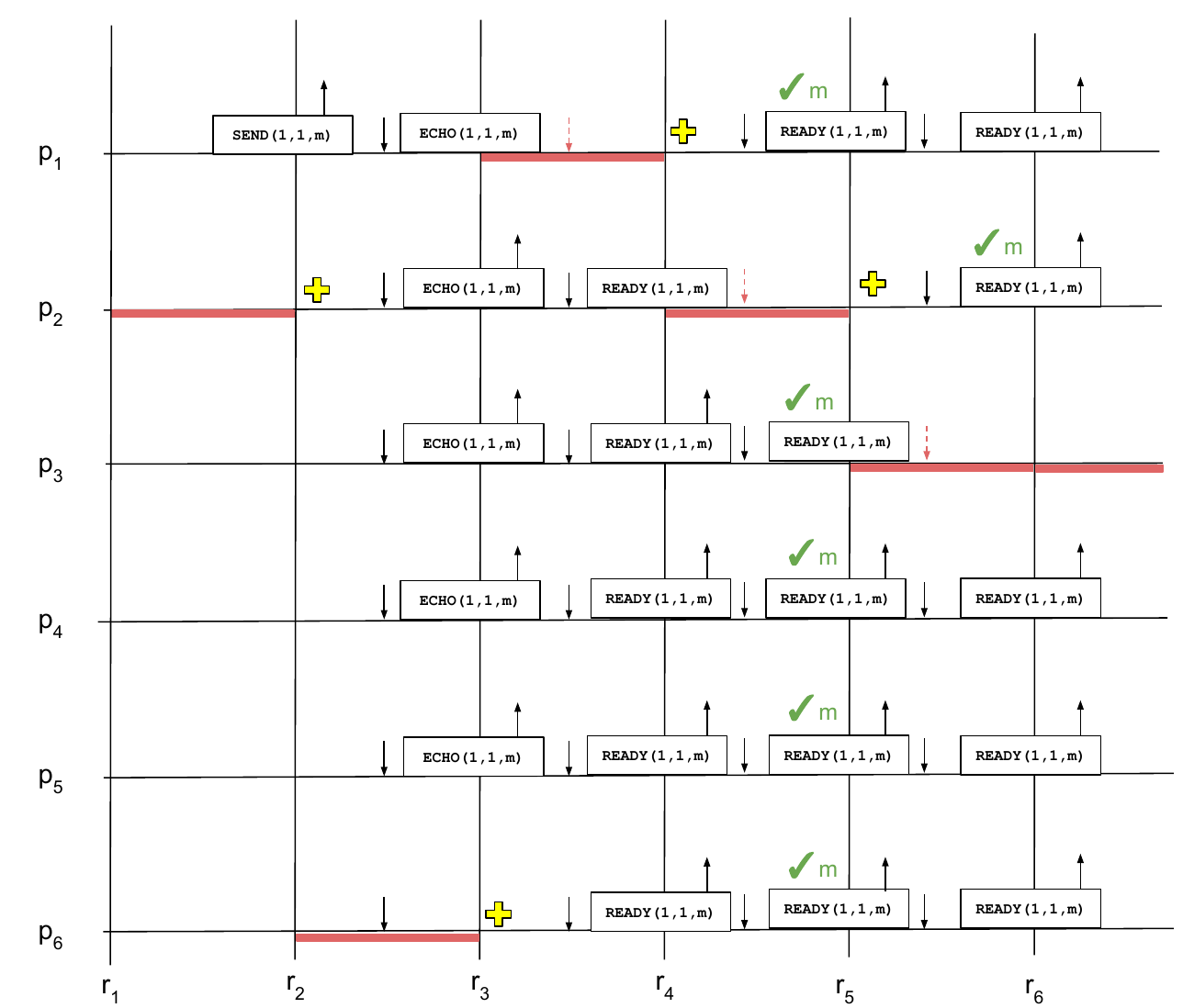}
\caption{An execution of $\mathcal{P}_{MBBC-RB}$ with a correct source and $f=1$.}
\label{fig:execcorrect}
\end{figure}

\begin{figure}
\centering
\includegraphics[width=.8\textwidth]{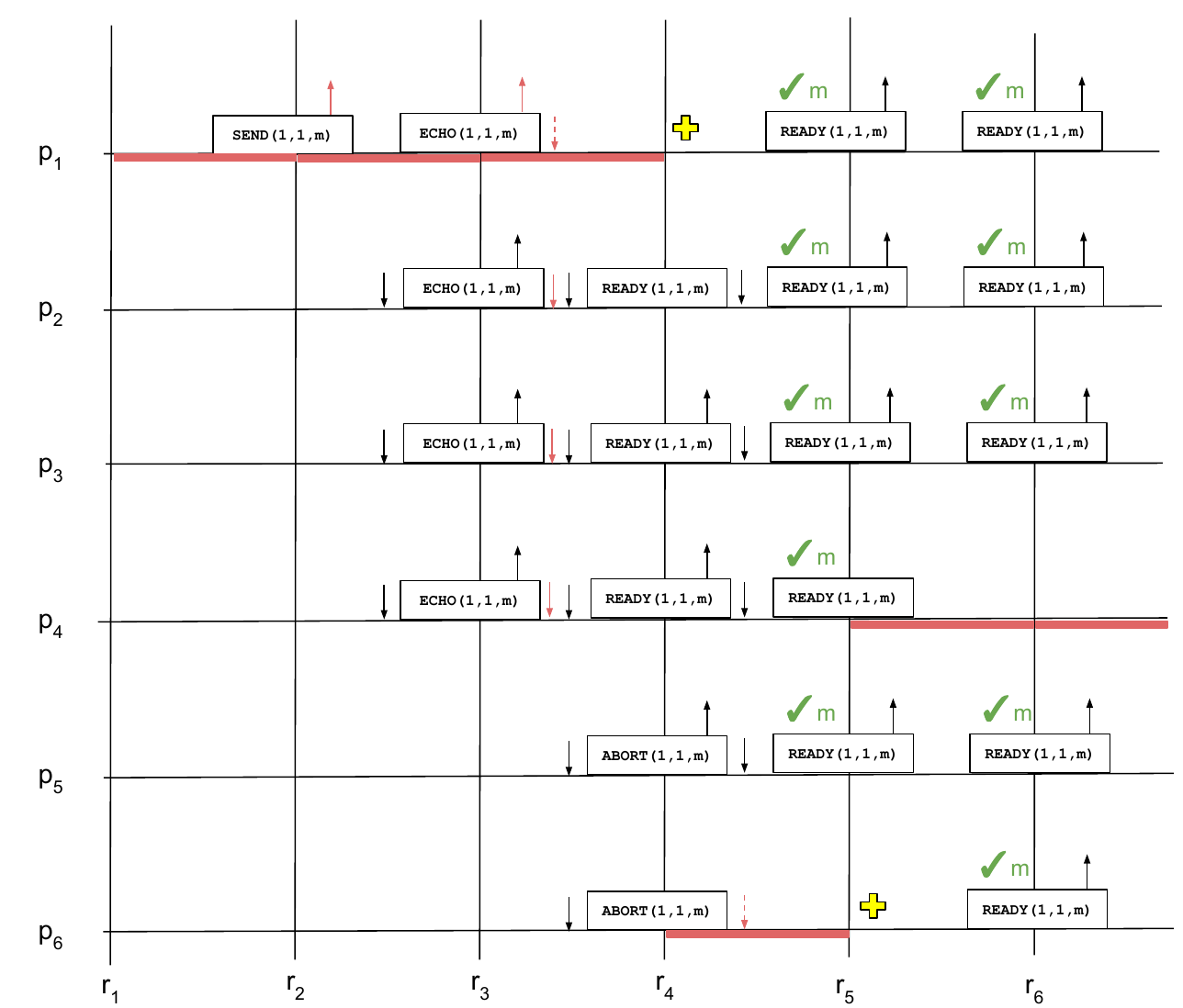}
\caption{An execution of $\mathcal{P}_{MBBC-RB}$ with a faulty source, $f=1$ and all infinitely often correct processes delivering.}
\label{fig:execbyz2}
\end{figure}

\begin{figure}
\centering
\includegraphics[width=.8\textwidth]{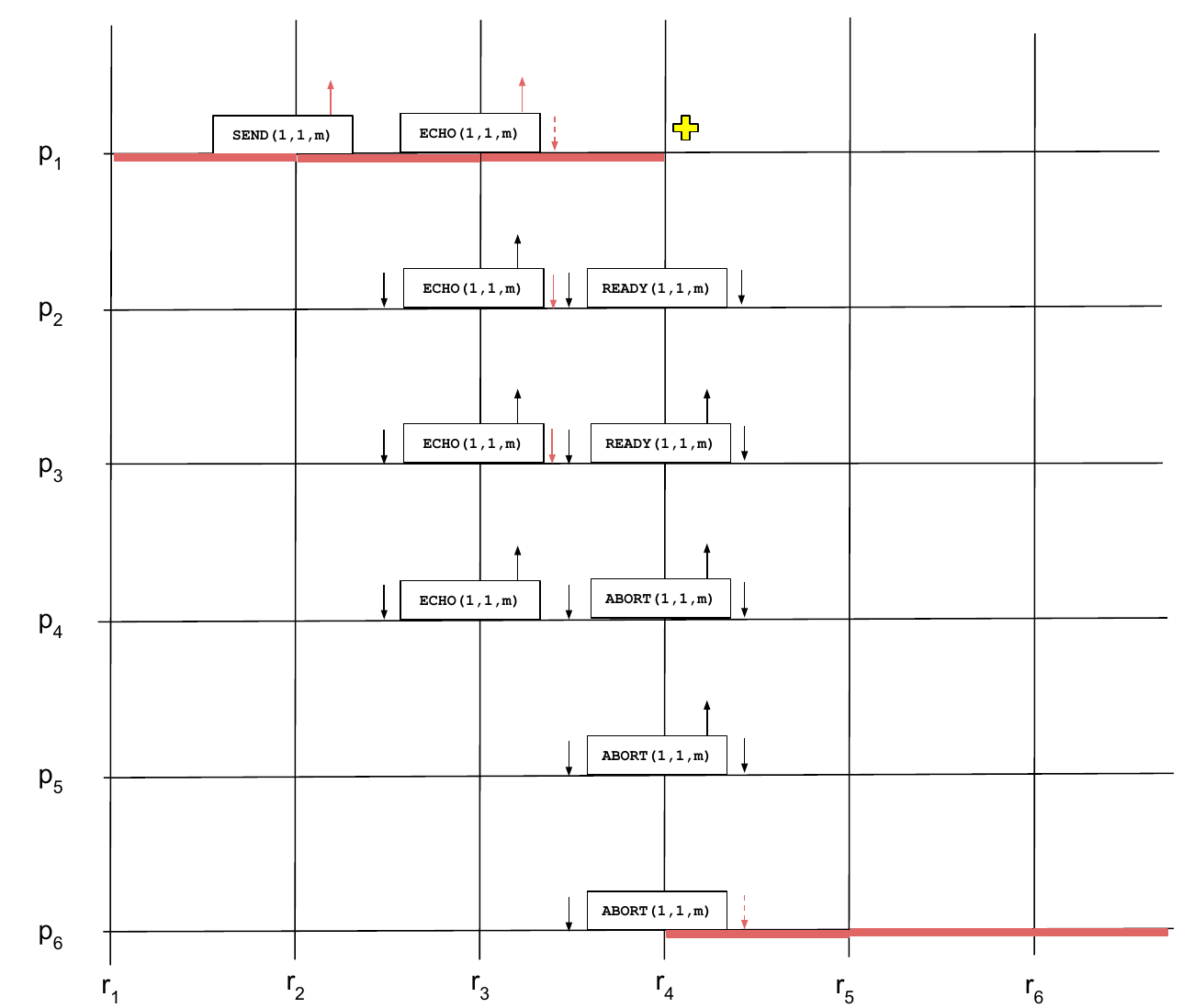}
\caption{An execution of $\mathcal{P}_{MBBC-RB}$ with a faulty source, $f=1$ and no infinitely often correct process delivering.}
\label{fig:execbyz1}
\end{figure}

    In the execution example in Figure \ref{fig:execcorrect}, the correct source $p_1$ starts the MBBC-Broadcast preparing the related \textsf{SEND} message in round $r_1$, that is \textsf{P2P}-sent to all processes in round $r_2$ ($\Delta_b=2$). Process $p_2$ is faulty in round $r_1$, then the mobile agent moves to process $p_6$ in round $r_2$.
All processes but $f$ are correct in round $r_2$, thus they receive the \textsf{SEND} message from $p_1$ and generate the related \textsf{ECHO} message. Such message is then \textsf{P2P}-sent to all peers by at least $n-2f$ processes during the \textit{send} phase in round $r_3$ (at most $f$ processes could have been faulty in round $r_2$, $p_6$ in our example, and at most $f$ processes could become faulty in round $r_3$, $p_1$ in our example where the mobile agent moves in round $r_3$). It follows that $n-f$ processes reach the quorum of \textsf{ECHO} messages generating the related \textsf{READY} message. Again, at least $n-2f$ processes are correct in round $r_4$, \textsf{P2P}-send the \textsf{READY} message and deliver the associated payload from $p_1$, $m$, during the \textit{compute} phase of the same round. 
The processes that were faulty in round $r_4$, $p_2$ in our example, deliver the message at the first round $r_k > r_4$ they get correct, because all processes that are correct in a round $r_j > r_4$ diffuse the associated \textsf{READY} message.
    
The only MBBC property that mobile agents may attempt to invalidate in a execution of $\mathcal{P}_{MBBC-RB}$ is the \textit{Agreement} property: the \textit{No duplication} is guaranteed by the \textit{if} statement at line \ref{algo:computedeliveryif2} in Algorithm \ref{alg:mbrbrb1} and both \textit{Validity} and \textit{Integrity} consider a correct source.
Any source must P2P-send a well-formed \textsf{SEND} message (i.e., with valid source id and round label) to make a correct process proceed in the protocol to deliver a payload $m$. If the \textsf{SEND} message is P2P-sent to all correct processes, then all $\Delta_c$-infinitely often correct processes will eventually deliver $m$, as shown in the previous execution, satisfying the MBBC specification.
It follows that a Byzantine source must not P2P-send the \textsf{SEND} message to some processes. 
This behavior has two possible outcomes in our protocol: either all correct processes MBBC-deliver the diffused message or no correct process does it.
Let us assume that the mobile agent commands $p_1$ to P2P-send the \textsf{SEND} message to $\lfloor (n-f)/2 \rfloor - f$ processes, in order to control which ones will proceed in the $\mathcal{P}_{MBBC-RB}$ protocol generating the \textsf{READY} message in round $r_4$.\\
In the execution depicted in Figure \ref{fig:execbyz2}, process $p_1$ is a faulty source that attempts to prevent the \textit{Agreement} property of MBBC from being satisfied. Specifically, it P2P-sends the \textsf{ECHO} message only to part of the processes, process $p_2$, $p_3$, and $p_4$, that reach the quorum required to generate the \textsf{READY} message. In this case, processes $p_5$ and $p_6$ generate the \textsf{ABORT} message but only $f$ of them, namely $p_5$, P2P-send it, thus blocking no correct process from proceeding in the MBBC-delivery of $m$ from $p_1$. Nonetheless, in this case more than $2f$ processes are correct and P2P-send the \textsf{READY} message in round $r_4$. It follows that all $\Delta_c$-infinitely often correct processes eventually deliver the associated payload $m$. \\
Differently from the previous example, in the execution in Figure \ref{fig:execbyz1} process $p_1$ sends the \textsf{ECHO} message to processes $p_2$ and $p_3$. It follows that all other correct processes, $p_4$, $p_5$, and $p_6$, generate the \textsf{ABORT} message. At most $f$ of them, process $p_6$ in the example, can be blocked from P2P-sending the \textsf{ABORT} message. It follows that more than $f$ processes diffuse to all correct ones the \textsf{ABORT} message and thus no process delivers the associated payload $m$. It follows that the specification is not violated in such execution.

\section{Additional details on $\mathcal{P}_{MBBC-RB}$ (Algorithm \ref{alg:mbrbrb1})}
\label{sec:variables}
We provide in this Appendix an additional detailed description of the $\mathcal{P}_{MBBC-RB}$ protocol and its variables defined in Algorithm \ref{alg:mbrbrb1} for the sake of completeness.
\begin{itemize}
\item The \textsf{Init} procedure initializes all the data structures and variables employed by the protocol. More in detail, it defines:
\begin{enumerate}
\item The \textit{To\_send} set variable to collect the messages (of any type) to P2P-send during the \textit{send} phase of a round;
\item The \textit{Sends} set to store the \textsf{SEND} messages received in a round;
\item The \textit{cured} boolean variable to keep track of the occurrence of the $\sf{O_{FFA}.cured}$ event;
\item The \textit{rc} integer variable to store the current round index;
\item The \textit{Echos}, \textit{Readys}, and \textit{Aborts} maps to collect, for every tuple $\langle s,m,r \rangle$ associated with single MBBC-instance, the identifier of the processes that P2P-send the \textsf{ECHO}, \textsf{READY}, and \textsf{ABORT} for such tuple respectively in the current round;
\item The \textit{RC} map that associates, for every process, the value of the round index that it P2P-sends in the current round. 
\end{enumerate}

\item The \textsf{Broadcast} procedure implements the \textsf{Broadcast} operation of \texttt{MBBC} by enqueueing the \textsf{SEND} message to P2P-send in set \textit{Sends}.

\item $\mathcal{P}_{MBBC-RB}$ is partitioned in three parts accordingly with the three phases assumed in the system model.

\item During the \textit{send} phase of a round, all the messages that have been enqueued to P2P-send in the previous round, stored in \textit{To\_send}, are discarded if the failure detector generated the \textsf{Cured} event in the current round, they are P2P-send to all processes otherwise.

\item The \textit{receive} phase of a round starts by wiping all maps data structures. The aim is to limit the capability of mobile agents to P2P-send spurious information only from $f$ processes in every round. Subsequently, all protocol's messages that have been P2P-received in the current round, \textsf{SEND}, \textsf{ECHO}, \textsf{READY}, and \textsf{ABORT}, are partitioned in the dedicated data structures.
The same occurs also for the messages exchanged to implement the fault-tolerant round counter: the P2P-received values are collected in the dedicated data structure.

\item The \textit{compute} phase analyzes all information received during the \textit{receive} phase and proceeds in the computation, MBBC-delivering messages and computing the protocol's messages to P2P-send in the subsequent round. 
It starts by wiping the data structure that collects the message to subsequently P2P-send and by updating the round index by majority (the $rc$ value may have been previously altered by an agent). 
Subsequently, if valid \texttt{SEND} message is received, then the related \texttt{ECHO} message is computed and enqueued to be P2P-send.
For all the \texttt{ECHO} messages received, if a process has received the one associated with a specific MBBC instance from a sufficient number of distinct processes (a quorum), then the \texttt{READY} message is computed and enqueued, if such a number is not sufficient but includes at least a currently correct process, then the \texttt{ABORT} message is generated.
If it is sure that the \texttt{ABORT} message associated with a MBBC instance has been sent by at least a correct process, then the received \texttt{READY} messages associated with the same instance are discarded in order to preserve the \textit{Agreement} property.
If a sufficient number of \texttt{READY} messages associated with a MBBC instance have been received and specific conditions are met, then the message is delivered.
These conditions are met at most once for every process, namely three rounds after the MBBC-Broadcast beginning or at the first subsequent round when a process get correct, for the minimum round label (making it irrelevant to the primitive).
Furthermore, if a sufficient amount of \texttt{READY} messages is received, then the same message is enqueued to P2P-send, in order to guarantee that processes that are faulty three rounds after the MBBC-broadcast will MBBC-deliver the associated payload when correct.
Finally, the round index is increased and its value is enqueued to P2P-send.
\end{itemize}

\end{appendices}

\bibliography{reference}

\begin{thebibliography}{10}

\bibitem{DBLP:conf/opodis/Abraham0X21}
Ittai Abraham, Ling Ren, and Zhuolun Xiang.
\newblock Good-case and bad-case latency of unauthenticated byzantine broadcast: {A} complete categorization.
\newblock In Quentin Bramas, Vincent Gramoli, and Alessia Milani, editors, {\em 25th International Conference on Principles of Distributed Systems, {OPODIS} 2021, December 13-15, 2021, Strasbourg, France}, volume 217 of {\em LIPIcs}, pages 5:1--5:20. Schloss Dagstuhl - Leibniz-Zentrum f{\"{u}}r Informatik, 2021.
\newblock \href {https://doi.org/10.4230/LIPIcs.OPODIS.2021.5} {\path{doi:10.4230/LIPIcs.OPODIS.2021.5}}.

\bibitem{DBLP:conf/podc/AlhaddadDD0VXZ22}
Nicolas Alhaddad, Sourav Das, Sisi Duan, Ling Ren, Mayank Varia, Zhuolun Xiang, and Haibin Zhang.
\newblock Balanced byzantine reliable broadcast with near-optimal communication and improved computation.
\newblock In Alessia Milani and Philipp Woelfel, editors, {\em {PODC} '22: {ACM} Symposium on Principles of Distributed Computing, Salerno, Italy, July 25 - 29, 2022}, pages 399--417. {ACM}, 2022.
\newblock \href {https://doi.org/10.1145/3519270.3538475} {\path{doi:10.1145/3519270.3538475}}.

\bibitem{DBLP:conf/dsn/BackesC03}
Michael Backes and Christian Cachin.
\newblock Reliable broadcast in a computational hybrid model with byzantine faults, crashes, and recoveries.
\newblock In {\em 2003 International Conference on Dependable Systems and Networks {(DSN} 2003), 22-25 June 2003, San Francisco, CA, USA, Proceedings}, pages 37--46. {IEEE} Computer Society, 2003.
\newblock \href {https://doi.org/10.1109/DSN.2003.1209914} {\path{doi:10.1109/DSN.2003.1209914}}.

\bibitem{banu2012improved}
Nazreen Banu, Samia Souissi, Taisuke Izumi, and Koichi Wada.
\newblock An improved byzantine agreement algorithm for synchronous systems with mobile faults.
\newblock {\em International Journal of Computer Applications}, 43(22):1--7, 2012.

\bibitem{DBLP:journals/jpdc/BoichatG05}
Romain Boichat and Rachid Guerraoui.
\newblock Reliable and total order broadcast in the crash-recovery model.
\newblock {\em J. Parallel Distributed Comput.}, 65(4):397--413, 2005.
\newblock \href {https://doi.org/10.1016/j.jpdc.2004.10.008} {\path{doi:10.1016/j.jpdc.2004.10.008}}.

\bibitem{DBLP:journals/tcs/BonnetDNP16}
Fran{\c{c}}ois Bonnet, Xavier D{\'{e}}fago, Thanh~Dang Nguyen, and Maria Potop{-}Butucaru.
\newblock Tight bound on mobile byzantine agreement.
\newblock {\em Theor. Comput. Sci.}, 609:361--373, 2016.
\newblock \href {https://doi.org/10.1016/j.tcs.2015.10.019} {\path{doi:10.1016/j.tcs.2015.10.019}}.

\bibitem{DBLP:conf/opodis/opodis23}
Silvia Bonomi, Giovanni Farina, and S\'ebastien Tixeuil'.
\newblock Reliable broadcast despite mobile byzantine faults.
\newblock In {\em 27th International Conference on Principles of Distributed Systems, {OPODIS} 2023, December 6-8, 2023, Tokyo, Japan}, 2023.

\bibitem{DBLP:conf/icdcn/BonomiPP16}
Silvia Bonomi, Antonella~Del Pozzo, and Maria Potop{-}Butucaru.
\newblock Tight self-stabilizing mobile byzantine-tolerant atomic register.
\newblock In {\em Proceedings of the 17th International Conference on Distributed Computing and Networking, Singapore, January 4-7, 2016}, pages 6:1--6:10. {ACM}, 2016.
\newblock \href {https://doi.org/10.1145/2833312.2833320} {\path{doi:10.1145/2833312.2833320}}.

\bibitem{DBLP:conf/icdcs/BonomiPPT16}
Silvia Bonomi, Antonella~Del Pozzo, Maria Potop{-}Butucaru, and S{\'{e}}bastien Tixeuil.
\newblock Approximate agreement under mobile byzantine faults.
\newblock In {\em 36th {IEEE} International Conference on Distributed Computing Systems, {ICDCS} 2016, Nara, Japan, June 27-30, 2016}, pages 727--728. {IEEE} Computer Society, 2016.
\newblock \href {https://doi.org/10.1109/ICDCS.2016.68} {\path{doi:10.1109/ICDCS.2016.68}}.

\bibitem{DBLP:conf/podc/BonomiPPT16}
Silvia Bonomi, Antonella~Del Pozzo, Maria Potop{-}Butucaru, and S{\'{e}}bastien Tixeuil.
\newblock Optimal mobile byzantine fault tolerant distributed storage: Extended abstract.
\newblock In George Giakkoupis, editor, {\em Proceedings of the 2016 {ACM} Symposium on Principles of Distributed Computing, {PODC} 2016, Chicago, IL, USA, July 25-28, 2016}, pages 269--278. {ACM}, 2016.
\newblock \href {https://doi.org/10.1145/2933057.2933100} {\path{doi:10.1145/2933057.2933100}}.

\bibitem{DBLP:conf/srds/BonomiPPT17}
Silvia Bonomi, Antonella~Del Pozzo, Maria Potop{-}Butucaru, and S{\'{e}}bastien Tixeuil.
\newblock Optimal storage under unsynchronized mobile byzantine faults.
\newblock In {\em 36th {IEEE} Symposium on Reliable Distributed Systems, {SRDS} 2017, Hong Kong, Hong Kong, September 26-29, 2017}, pages 154--163. {IEEE} Computer Society, 2017.
\newblock \href {https://doi.org/10.1109/SRDS.2017.20} {\path{doi:10.1109/SRDS.2017.20}}.

\bibitem{DBLP:journals/iandc/Bracha87}
Gabriel Bracha.
\newblock Asynchronous byzantine agreement protocols.
\newblock {\em Inf. Comput.}, 75(2):130--143, 1987.
\newblock \href {https://doi.org/10.1016/0890-5401(87)90054-X} {\path{doi:10.1016/0890-5401(87)90054-X}}.

\bibitem{DBLP:conf/ftcs/BuhrmanGH95}
Harry Buhrman, Juan~A. Garay, and Jaap{-}Henk Hoepman.
\newblock Optimal resiliency against mobile faults.
\newblock In {\em Digest of Papers: FTCS-25, The Twenty-Fifth International Symposium on Fault-Tolerant Computing, Pasadena, California, USA, June 27-30, 1995}, pages 83--88. {IEEE} Computer Society, 1995.
\newblock \href {https://doi.org/10.1109/FTCS.1995.466995} {\path{doi:10.1109/FTCS.1995.466995}}.

\bibitem{DBLP:books/daglib/0025983}
Christian Cachin, Rachid Guerraoui, and Lu{\'{\i}}s E.~T. Rodrigues.
\newblock {\em Introduction to Reliable and Secure Distributed Programming {(2.} ed.)}.
\newblock Springer, 2011.
\newblock \href {https://doi.org/10.1007/978-3-642-15260-3} {\path{doi:10.1007/978-3-642-15260-3}}.

\bibitem{DBLP:journals/cacm/Dijkstra74}
Edsger~W. Dijkstra.
\newblock Self-stabilizing systems in spite of distributed control.
\newblock {\em Commun. {ACM}}, 17(11):643--644, 1974.
\newblock \href {https://doi.org/10.1145/361179.361202} {\path{doi:10.1145/361179.361202}}.

\bibitem{DBLP:books/mit/Dolev2000}
Shlomi Dolev.
\newblock {\em Self-Stabilization}.
\newblock {MIT} Press, 2000.
\newblock URL: \url{http://www.cs.bgu.ac.il/\%7Edolev/book/book.html}.

\bibitem{DBLP:conf/wdag/Garay94}
Juan~A. Garay.
\newblock Reaching (and maintaining) agreement in the presence of mobile faults (extended abstract).
\newblock In Gerard Tel and Paul M.~B. Vit{\'{a}}nyi, editors, {\em Distributed Algorithms, 8th International Workshop, {WDAG} '94, Terschelling, The Netherlands, September 29 - October 1, 1994, Proceedings}, volume 857 of {\em Lecture Notes in Computer Science}, pages 253--264. Springer, 1994.
\newblock \href {https://doi.org/10.1007/BFb0020438} {\path{doi:10.1007/BFb0020438}}.

\bibitem{DBLP:conf/opodis/GuerraouiKKPST20}
Rachid Guerraoui, Jovan Komatovic, Petr Kuznetsov, Yvonne{-}Anne Pignolet, Dragos{-}Adrian Seredinschi, and Andrei Tonkikh.
\newblock Dynamic byzantine reliable broadcast.
\newblock In Quentin Bramas, Rotem Oshman, and Paolo Romano, editors, {\em 24th International Conference on Principles of Distributed Systems, {OPODIS} 2020, December 14-16, 2020, Strasbourg, France (Virtual Conference)}, volume 184 of {\em LIPIcs}, pages 23:1--23:18. Schloss Dagstuhl - Leibniz-Zentrum f{\"{u}}r Informatik, 2020.
\newblock \href {https://doi.org/10.4230/LIPIcs.OPODIS.2020.23} {\path{doi:10.4230/LIPIcs.OPODIS.2020.23}}.

\bibitem{DBLP:conf/wdag/GuerraouiKMPS19}
Rachid Guerraoui, Petr Kuznetsov, Matteo Monti, Matej Pavlovic, and Dragos{-}Adrian Seredinschi.
\newblock Scalable byzantine reliable broadcast.
\newblock In Jukka Suomela, editor, {\em 33rd International Symposium on Distributed Computing, {DISC} 2019, October 14-18, 2019, Budapest, Hungary}, volume 146 of {\em LIPIcs}, pages 22:1--22:16. Schloss Dagstuhl - Leibniz-Zentrum f{\"{u}}r Informatik, 2019.
\newblock \href {https://doi.org/10.4230/LIPIcs.DISC.2019.22} {\path{doi:10.4230/LIPIcs.DISC.2019.22}}.

\bibitem{DBLP:journals/ppl/ImbsR16}
Damien Imbs and Michel Raynal.
\newblock Trading off \emph{t}-resilience for efficiency in asynchronous byzantine reliable broadcast.
\newblock {\em Parallel Process. Lett.}, 26(4):1650017:1--1650017:8, 2016.
\newblock \href {https://doi.org/10.1142/S0129626416500171} {\path{doi:10.1142/S0129626416500171}}.

\bibitem{DBLP:conf/issre/KoutrasP20}
Vasilis~P. Koutras and Agapios~N. Platis.
\newblock Chapter 3: Software rejuvenation: Key concepts and granularity.
\newblock In {\em 2020 {IEEE} International Symposium on Software Reliability Engineering Workshops, {ISSRE} Workshops, Coimbra, Portugal, October 12-15, 2020}, pages 321--322. {IEEE}, 2020.
\newblock \href {https://doi.org/10.1109/ISSREW51248.2020.00092} {\path{doi:10.1109/ISSREW51248.2020.00092}}.

\bibitem{DBLP:journals/access/LiYWW22}
Jing Li, Tianming Yu, Ye~Wang, and Roger Wattenhofer.
\newblock Dynamic byzantine broadcast in asynchronous message-passing systems.
\newblock {\em {IEEE} Access}, 10:91372--91384, 2022.
\newblock \href {https://doi.org/10.1109/ACCESS.2022.3202627} {\path{doi:10.1109/ACCESS.2022.3202627}}.

\bibitem{DBLP:journals/jnca/LiaoLLT13}
Hung{-}Jen Liao, Chun{-}Hung~Richard Lin, Ying{-}Chih Lin, and Kuang{-}Yuan Tung.
\newblock Intrusion detection system: {A} comprehensive review.
\newblock {\em J. Netw. Comput. Appl.}, 36(1):16--24, 2013.
\newblock \href {https://doi.org/10.1016/j.jnca.2012.09.004} {\path{doi:10.1016/j.jnca.2012.09.004}}.

\bibitem{DBLP:conf/podc/OstrovskyY91}
Rafail Ostrovsky and Moti Yung.
\newblock How to withstand mobile virus attacks (extended abstract).
\newblock In Luigi Logrippo, editor, {\em Proceedings of the Tenth Annual {ACM} Symposium on Principles of Distributed Computing, Montreal, Quebec, Canada, August 19-21, 1991}, pages 51--59. {ACM}, 1991.
\newblock \href {https://doi.org/10.1145/112600.112605} {\path{doi:10.1145/112600.112605}}.

\bibitem{DBLP:books/sp/Raynal18}
Michel Raynal.
\newblock {\em Fault-Tolerant Message-Passing Distributed Systems - An Algorithmic Approach}.
\newblock Springer, 2018.
\newblock \href {https://doi.org/10.1007/978-3-319-94141-7} {\path{doi:10.1007/978-3-319-94141-7}}.

\bibitem{DBLP:journals/ppl/Raynal21}
Michel Raynal.
\newblock On the versatility of bracha's byzantine reliable broadcast algorithm.
\newblock {\em Parallel Process. Lett.}, 31(3):2150006:1--2150006:9, 2021.
\newblock \href {https://doi.org/10.1142/S0129626421500067} {\path{doi:10.1142/S0129626421500067}}.

\bibitem{DBLP:journals/iandc/Reischuk85}
R{\"{u}}diger Reischuk.
\newblock A new solution for the byzantine generals problem.
\newblock {\em Inf. Control.}, 64(1-3):23--42, 1985.
\newblock \href {https://doi.org/10.1016/S0019-9958(85)80042-5} {\path{doi:10.1016/S0019-9958(85)80042-5}}.

\bibitem{DBLP:conf/icdcs/RodriguesR00}
Lu{\'{\i}}s E.~T. Rodrigues and Michel Raynal.
\newblock Atomic broadcast in asynchronous crash-recovery distributed systems.
\newblock In {\em Proceedings of the 20th International Conference on Distributed Computing Systems, Taipei, Taiwan, April 10-13, 2000}, pages 288--295. {IEEE} Computer Society, 2000.
\newblock \href {https://doi.org/10.1109/ICDCS.2000.840941} {\path{doi:10.1109/ICDCS.2000.840941}}.

\bibitem{DBLP:conf/trustcom/SabtAB15}
Mohamed Sabt, Mohammed Achemlal, and Abdelmadjid Bouabdallah.
\newblock Trusted execution environment: What it is, and what it is not.
\newblock In {\em 2015 {IEEE} TrustCom/BigDataSE/ISPA, Helsinki, Finland, August 20-22, 2015, Volume 1}, pages 57--64. {IEEE}, 2015.
\newblock \href {https://doi.org/10.1109/Trustcom.2015.357} {\path{doi:10.1109/Trustcom.2015.357}}.

\bibitem{DBLP:conf/nca/SakavalasT18}
Dimitris Sakavalas and Lewis Tseng.
\newblock Delivery delay and mobile faults.
\newblock In {\em 17th {IEEE} International Symposium on Network Computing and Applications, {NCA} 2018, Cambridge, MA, USA, November 1-3, 2018}, pages 1--8. {IEEE}, 2018.
\newblock \href {https://doi.org/10.1109/NCA.2018.8548345} {\path{doi:10.1109/NCA.2018.8548345}}.

\bibitem{DBLP:conf/opodis/SasakiYKY13}
Toru Sasaki, Yukiko Yamauchi, Shuji Kijima, and Masafumi Yamashita.
\newblock Mobile byzantine agreement on arbitrary network.
\newblock In Roberto Baldoni, Nicolas Nisse, and Maarten van Steen, editors, {\em Principles of Distributed Systems - 17th International Conference, {OPODIS} 2013, Nice, France, December 16-18, 2013. Proceedings}, volume 8304 of {\em Lecture Notes in Computer Science}, pages 236--250. Springer, 2013.
\newblock \href {https://doi.org/10.1007/978-3-319-03850-6\_17} {\path{doi:10.1007/978-3-319-03850-6\_17}}.

\bibitem{DBLP:conf/sss/Tseng17}
Lewis Tseng.
\newblock An improved approximate consensus algorithm in the presence of mobile faults.
\newblock In Paul~G. Spirakis and Philippas Tsigas, editors, {\em Stabilization, Safety, and Security of Distributed Systems - 19th International Symposium, {SSS} 2017, Boston, MA, USA, November 5-8, 2017, Proceedings}, volume 10616 of {\em Lecture Notes in Computer Science}, pages 109--125. Springer, 2017.
\newblock \href {https://doi.org/10.1007/978-3-319-69084-1\_8} {\path{doi:10.1007/978-3-319-69084-1\_8}}.

\end{thebibliography}

\end{document}